\documentclass[useAMS,usenatbib, referee]{biom}

\usepackage[figuresright]{rotating}
\usepackage{amsmath,comment}
\usepackage{subcaption}
\newcommand{\norm}[1]{\left\lVert#1\right\rVert}
\usepackage{xcolor}
\usepackage{mathtools}
\usepackage{amsmath}
\usepackage{booktabs}
\usepackage{amsmath,comment}
\usepackage{mathrsfs}
\usepackage{amssymb}
\usepackage{enumitem}
\usepackage[colorlinks,citecolor=blue,urlcolor=blue,linkcolor=blue]{hyperref}

\DeclareMathOperator{\argmax}{argmax}
\title{Nonparametric estimation of the causal effect of a stochastic threshold-based intervention}
\author{Lars van der Laan$^{1,3*}$\email{lvdlaan@uw.edu}, Wenbo Zhang$^2$, Peter B. Gilbert$^{2,3}$ \\
$^{1}$Department of Statistics, University of Washington, Seattle, Washington, 98109, U.S.A. \\
$^{2}$Department of Biostatistics \\ University of Washington, Seattle, Washington, 98109, U.S.A.\\
$^{3}$Vaccine and Infectious Disease and Public Health Sciences Divisions, \\
Fred Hutchinson Cancer Research Center, Seattle, Washington, 98109, U.S.A. \\
}\date{}          

\newcommand\independent{\protect\mathpalette{\protect\independenT}{\perp}}
\def\independenT#1#2{\mathrel{\rlap{$#1#2$}\mkern2mu{#1#2}}}

\usepackage{amssymb}

\begin{document}

\begin{abstract}
Identifying a biomarker or treatment-dose threshold that marks a specified level of risk is an important problem, especially in clinical trials. In view of this goal, we consider a covariate-adjusted threshold-based interventional estimand, which happens to equal the binary treatment-specific mean estimand from the causal inference literature obtained by dichotomizing the continuous biomarker or treatment as above or below a threshold. The unadjusted version of this estimand was considered in Donovan et al. (2019). Expanding upon Stitelmen et al. (2010), we show that this estimand, under conditions, identifies the expected outcome of a stochastic intervention that sets the treatment dose of all participants above the threshold. We propose a novel nonparametric efficient estimator for the covariate-adjusted threshold-response function for the case of informative outcome missingness, which utilizes machine learning and Targeted Minimum-Loss Estimation (TMLE). We prove the estimator is efficient and characterize its asymptotic distribution and robustness properties. Construction of simultaneous 95\% confidence bands for the threshold-specific estimand across a set of thresholds is discussed. In the supplementary information, we discuss how to adjust our estimator when the biomarker is missing-at-random, as occurs in clinical trials with biased sampling designs, using inverse-probability-weighting. Efficiency and bias-reduction of the proposed estimator are assessed in simulations. The methods are employed to estimate neutralizing antibody thresholds for virologically confirmed dengue risk in the CYD14 and CYD15 dengue vaccine trials. 

\end{abstract}

\maketitle

\section{Introduction and previous work}

 In clinical trials, it is often of interest to identify a biomarker that is predictive of a clinical outcome of interest. In particular, in vaccine efficacy trials, it is of interest to find so-called 
 correlates of risk, such as neutralizing antibody titer, that are indicative of the risk of acquisition of disease.
If additional analyses show that such a biomarker correlate is also a valid surrogate endpoint, then it can be used to predict vaccine efficacy of new vaccines by only analyzing the biomarker, as opposed to observing the clinical endpoints. Such analysis generally requires only hundreds, rather than thousands, of participants in the vaccine study, and allows for vaccine efficacy to be assessed efficiently in terms of both economic resources and time. 
As an intermediate step toward meeting this objective, it is of interest to determine a threshold value of an immune-response biomarker that predicts a low risk of disease. While there is a large literature of statistical methods for estimating correlates of risk and protection (e.g., Chan et al., 2002; Siber et al., 2007; Callegaro et al., 2019),\nocite{Chan2002,Siberetal2007,CallegaroTibaldi2019} prominent methods rely on parametric assumptions, which suffer from misspecification when their strict assumptions are not met. Donovan, Hudgens and Gilbert (2019)\nocite{DonovanHudgensGilbert2019} proposed a nonparametric minimum-loss estimator for the so-called unadjusted threshold-response function  $E[Y \mid  A \geq v]$ with $Y$ the outcome and $A$ the biomarker, and used the nonparametric bootstrap for inferences. 
The threshold-response function can be viewed as a dose-response-like curve, which maps each threshold to the expected outcome given the biomarker is above that threshold. Unlike the nonparametric dose-response curve, the threshold-response function is $\sqrt{n}$-estimable, allowing one to construct efficient estimators and $95\%$ (simultaneous) confidence bands using standard techniques from semiparametric efficiency theory. 

Even in randomized trials, the immune-response biomarker, outcome and outcome missingness mechanisms are generally not randomized itself and thus may be confounded by baseline covariates. In order to have any hope of estimating a causal effect that is due to a threshold-based intervention on the biomarker, covariate adjustment is generally required. Typical techniques such as covariate-stratified estimation, as mentioned in Donovan et al., can lead to confounding bias due to discretization and perform worse as the number of stratification covariates increase. Thus, there is a need to generalize the estimator and estimand proposed by Donovan et al. in a way that allows for flexible and adaptive covariate adjustment. Efficient-influence-function-based causal inference methods (Bickel et al., 1993; van der Laan, Robins, 2003; van der Laan, Rose, 2011)\nocite{unified} provide general tools for efficient nonparametric estimation of such covariate-adjusted estimands, allowing for the use of machine-learning. Specifically, in this manuscript, we employ the targeted learning methodology (van der Laan, Rose, 2011)\nocite{vanderLaanRose2011}, which provides a general template for constructing efficient substitution estimators.  

We extend the previous work of Donovan et al. by allowing for the nonparametric adjustment of arbitrary baseline covariates in the presence of possibly informative outcome missingness. In addition to (1) proposing an efficient targeted minimum-loss estimator (TMLE) that generalizes the estimator of Donovan et al. for the covariate-adjusted threshold-response estimand, $E_W \left[ E  [Y \mid A \geq v,\, W \right]$, in the presence of possibly informative outcome missingness, further novel additions include: (2) Establishing rigorous causal identification results for the threshold-response parameter as a stochastic intervention; (3) Proving asymptotic efficiency and robustness properties of the proposed TMLE; (4) Providing simultaneous confidence bands. Following the treatment by Donovan et al., we discuss in Web Appendix E how to adjust our method when the biomarker is missing-at-random using inverse-probability weighting. This allows our method to be applied in studies with biased sampling designs.  

The covariate-adjusted threshold-response estimand considered in this manuscript has been considered earlier in the literature.  Notably, the threshold-response estimand with no outcome missingness is equal to the G-computation-based binary treatment-specific mean estimand (Robins et al., 1994; van der Laan, Robins, 2003; van der Laan, Rose, 2011)\nocite{robinsAIPW} where the continuous biomarker or treatment is  dichotomized into an indicator of being above or below a threshold. The idea of dichotomizing a continuous biomarker, exposure or treatment variable as an indicator of being above or below a threshold and then applying causal inference methods for binary treatment effects is not a new idea and is used, for instance, in Taubman et al. (2009). Moreover, in an unpublished manuscript, Stitelman et al. (2010)\nocite{Stitelman2010TheIO} discuss changes in causal interpretation and confounding bias due to discretizing continuous treatments and then applying binary or categorical treatment causal inference methods like the AIPW or TMLE estimator for the treatment-specific mean (van der Laan, Robins; 2003; Bang and Robins, 2005; van der Laan, Rose, 2011). In Sections 4 and 5 of their manuscript, the authors show that the binary-treatment-specific mean estimand based on dichotomizing a continuous treatment as above or below a threshold can be viewed as a stochastic intervention. In this manuscript, we provide the following novel additions to the aforementioned authors' work: (1) We show that in the case of no outcome missingness there is no loss in statistical efficiency by applying binary treatment causal inference methods to estimate the threshold-response estimand; (2) We show that when there is outcome missingness the binary-treatment-based estimators are statistically inefficient for the threshold-response estimand, and therefore there is a need to develop a new TMLE that is fully efficient. Moreover, if the outcome missingness is informed by the biomarker value then the binary treatment estimators can be inconsistent due to not fully adjusting for confounding between the continuous treatment and outcome missingness. 
 
The structure of the article is as follows. In Sections 2 and 3, we define the data structure, parameter of interest, and causal assumptions needed for the threshold-response estimand to be identified and interpreted as a stochastic intervention. In Section 4, the efficient influence function of the target parameter is given. In Section 5, we present a novel sequential-regression-based TMLE (srTMLE) and establish its theoretical properties. We also discuss how to construct simultaneous confidence intervals and in the supplementary information we discuss how to adjust the srTMLE when the biomarker variable is missing-at-random. In Section 6, we investigate through simulations the efficiency and bias-reduction of the proposed srTMLE relative to the Donovan et al. estimator and the inefficient binary-treatment specific-mean TMLE (binTMLE). In Section 7, we apply the new method to the CYD14 and CYD15 dengue vaccine trials.

\section{Notation, data structure, and estimand of interest}

Consider a study where we observe the $n$ iid realizations $O_i$ of the random variable $O = (W, A, \Delta, \Delta Y) \sim P_0$ where $P_0$ is the data-generating distribution. Here, $W \in \mathbb{R}^d$ represents baseline variables, $A \in \mathbb{R}$ is a continuous biomarker of interest measured during follow-up, $\Delta$ is an outcome missingness indicator that takes the value $1$ if the outcome is observed, and $ Y$ is a binary outcome variable of interest measured at end-of-study. We use notation $\Delta Y$ to denote that $Y$ is observed if and only if $\Delta = 1$. For example, in the vaccine trial setting, $A$ may be an immune-response biomarker measured some time during the trial, and $Y \in \{0,1\}$ may be the binary variable that takes the value 1 if the participant acquired the disease study endpoint by the end of the follow-up period, and 0 otherwise. Let $v \in \mathbb{R}$ be a given threshold in the support of $A$ such that $P_0(A \geq v \mid W) \geq \delta$ a.e. $W$ for some $\delta > 0$, and define the dichotomized biomarker $D_v = 1(A \geq v)$ and coarsened data-structure $O_v := (W, D_v, \Delta, \Delta Y)$, which will be referenced throughout this manuscript.
We assume that $P_0$ is contained in a nonparametric statistical model and let $\norm{\,\cdot\,}$ denote the $L^2(P_0)$ norm. We will denote $P_W$ as the marginal distribution of $W$, $P_{A \mid W}$ as the conditional distribution of $A \mid W$, $Q(A,W) = E_P[Y=1 \mid A, \, W,\, \Delta =1]$, $Q_v(W) = E_{P_{A \mid W}}[E_P[Y=1 \mid A, \,W, \,\Delta =1] \mid   A \geq v,\, W]$, $g_v(W) = P(A \geq v \mid W)$. Let $Q_0$, $Q_{0,v}$, $g_{0,v}$, $P_{0,W}$, $P_{0,A \mid W}$ correspond with $P = P_0$. Throughout, we will abbreviate $E_{P_0},  E_{P_{0,A \mid W}}, E_{P_{0,W}}$ as $E_0, E_{0,A \mid W}, E_{0,W}$. Occasionally, we will use the empirical process notation: $P_0 f = E_{0} f(O)$  and $P_n f = \frac{1}{n} \sum_{i=1}^n f(O_i)$ for a function $o \mapsto f(o)$. 

\noindent Our estimand of interest is
\begin{equation}
\Psi_{v,0}^{adj}  = E_{0,W}\left( E_0[Y \mid  A \geq v,\,W] \right)  = E_{0,W} \left( E_{0,A\mid W}   \left[ E_0 \left[Y \mid A,\, W \right] \mid A\geq v,\, W \right] \right). \label{equation::fulldataestimand}
\end{equation}
To ensure the estimand is well-defined and identified from the observed data-generating distribution, we will make the following assumptions on $P_0$.
\begin{enumerate}[series=conditions,label=({A}{\arabic*})]
\item There exists $\delta > 0$ s.t. $P_0(A \geq v  \mid W) > \delta$ and $P_0(\Delta = 1 \mid A,\, W) > \delta$ a.e. $A, \, W$
\item $\Delta \independent Y \mid A,\, W$ ($Y$ is missing-at-random).
\end{enumerate}

Assumption (A1) consists of standard overlap/positivity conditions that ensure the estimands are well-defined. Specifically, the assumption ensures that there is a positive probability of observing a biomarker value above the threshold within all strata of $W$ and a positive probability of observing the outcome $Y$ within all strata of $(W,A)$. Assumption (A2) is a missing-at-random assumption that ensures that the outcome missingness is noninformative conditional on $W$ and $A$ and is also commonly made for missing-data problems (van der Laan, Robins, 2003).

 \noindent Motivated by the latter form of the estimand given in Equation (\ref{equation::fulldataestimand}), we define the target parameter $\Psi_v^{adj}: \mathcal{M} \mapsto \mathbb{R}$ 
\begin{equation}
    \Psi_v^{adj}(P) = E_{P_W}\left( E_{P_{A\mid W}} \left[ E_P \left[Y \mid A,  \, W,   \, \Delta =1 \right] \mid A \geq v,  \, W  \right]  \right) \label{equation::param}.
\end{equation}
Under assumptions (A1) and (A2), the observed-data estimand $\Psi_v^{adj}(P_0)$ identifies the estimand given in Equation (\ref{equation::fulldataestimand}). We will call this estimand the (adjusted) \textit{threshold-response} at the threshold $v$, and we will call the map $v \mapsto \Psi_v^{adj}(P_0)$ the \textit{threshold-response function}. In the special case where the missingness $\Delta$ is not informed by the marker $A$ conditional on $W$ and $A \geq v$, i.e. $Y \independent \Delta \mid W,\, A \geq v$, Equation (\ref{equation::param}) reduces to the well-known binary treatment-specific mean estimand: $\Psi_v^{adj}(P) = E_{P_W} \left[E_P\left[Y \mid A \geq v, \, W,  \,\Delta = 1 \right] \right] = E_{P_W} \left[E_P\left[Y \mid D_v = 1, \, W,  \,\Delta = 1 \right] \right]$. Notably, the simpler form of the estimand is determined by the data-generating distribution of the coarsened data-structure $O_v=(W, D_v, \Delta, \Delta Y)$. We will refer to this case as the quasi-informative missingness case, which includes the case of no outcome missingness. We define the unadjusted threshold-response function $\Psi_v^{unadj}(P) = E_P[Y \mid A \geq v, \,\Delta = 1].$ \noindent In the case where $W$ is independent of $A$, the unadjusted and adjusted threshold response functions are equal.

\section{Causal interpretation and identification of the threshold-response function as a stochastic intervention}

Since we have already presented conditions under which the outcome missingness estimand implied by Equation (\ref{equation::param}) identifies the estimand given in Equation (\ref{equation::fulldataestimand}), we restrict ourselves to the case of no outcome missingness and consider the estimand given in Equation (\ref{equation::fulldataestimand}). Informally and under conditions, the threshold-response estimand can be viewed as the expected outcome under an intervention through the continuous biomarker $A$ that sets the dichotomized biomarker or binary ``treatment" variable $D_v = 1(A \geq v)$ to $1$ for a given individual. Moreover, this intervention is such that an individual with $A \geq v$ is not intervened upon. From this point of view, the threshold-response estimand is quite similar to standard interventional estimands for binary treatments (Robins et al., 1994; van der Laan, Robins, 2003; van der Laan, Rose, 2011). However, the intervention mechanism associated with the threshold-response estimand acts through the continuous marker $A$, and, in fact, there are many different interventions on individuals with $A<v$ preintervention that set $D_v \equiv 1$. Since the distribution of the outcome $Y$ generally depends on the continuous values of $A$ and not only the binary indicator $D_v$, different interventions that set $D_v \equiv 1$ can lead to different expected counterfactual outcomes (Stitelman et al., 2010).

To rigorously formulate the notion of an intervention, we consider the following causal model. Let $U_W, U_A, U_Y$ be exogenous random variables, and let $f_W, f_A, f_Y$ be deterministic functions. We define the underlying causal data structure to be $O_{causal} = (W, A, Y, U_W, U_A , U_Y) \sim P_{0,c }$.  We assume that $O = (W,A,Y)$ is causally generated by the following nonparametric structural equations model (NPSEM) (Pearl, 2009)\nocite{Pearl2009},
$$W = f_W(U_W), \, A = f_A(W, U_A),\, Y = f_{Y}(A, W, U_Y).$$
Following the formulation given in Pearl (2009) and the notation of D\'iaz et al. (2021), we define an intervention on $A$ as a rule that maps a realization $(w,a)$ of $(W,A)$ to a possibly randomized biomarker value $d_0(w,a) \in \text{support}(A)$ that may depend on $P_0$. For a given observation $O = (W,A,Y)$, the counterfactual outcome associated with the intervention $d_0$ is $Y_{d_0} := f_Y(d_0(W,A), W, U_Y)$. 
If $A \in \{0,1\}$ were binary then $d_{v=1}(a) := 1(a \geq 1) \equiv 1(a = 1)$ is a valid intervention that corresponds with the standard binary treatment intervention and $Y_1 := Y_{d_{v=1}}$ would be the familiar counterfactual outcome associated with the treatment assignment $A \equiv 1$. 


The threshold-response estimand given in Equation (\ref{equation::fulldataestimand}) can be viewed as the expected outcome under a randomized intervention $\widetilde{d}_{0,v}$ on $A$ that satisfies $1(\widetilde{d}_{0,v}(A,W) \geq v) =1$ and is partially stochastic conditional on $(A,W)$. Specifically,  $\widetilde{d}_{0,v}$ is given by
$\widetilde{d}_{0,v}(w,a) := 1(a < v) \cdot  F_{0,v}^{-1}(Z \mid w) + 1(a\geq v) \cdot a$
where $F_{0,v}$ is the conditional CDF $s \mapsto P_0(A \leq s \mid W=w,\, A \geq v)$ and $Z \sim P_Z$ is a randomizer that is uniformly distributed on $[0,1]$ and independent of $O$. Noting that $F_{0,v}^{-1}(Z \mid w)$ is exactly an independent draw from the conditional CDF $F_{0,v}$, we are guaranteed that the assigned interventional marker value is always above the threshold $v$. 
\begin{remark}
In the special case where $W$ is discretely valued, the above interventional value $\widetilde{d}_{0,v}(w,a)$ can be computed in a more intuitive way as follows. Consider an infinitely-large target population consisting of realizations of $(W,A,Y) \sim P_0$. For a given individual $ O= (W,A,Y)$ with $(W=w,\,A=a)$ and $a < v$, uniformly-at-random match this individual with a member $O_{match,v}(w)$ of the population with the same baseline covariates $W_{match,v}=w$ but with biomarker value $A_{match,v}(w) \geq v$. After matching, assign the interventional biomarker value of the individual $O$ to the observed biomarker value of the matching individual $O_{match,v}(w)$. For individuals in the sample with $A \geq v$, no intervention is performed. In this case, the random interventional value can be written as $\widetilde{d}_{0,v}(A,W) = A \cdot 1(A \geq v) + A_{match,v}(W)\cdot 1(A < v)$.
\end{remark}

Back to the general case, the expected outcome under the intervention $\widetilde{d}_{0,v}$ is given by
$\Psi_{v,causal}(P_{0,c}) := E_{P_{0,c} \times P_Z}[Y_{\widetilde{d}_{0,v}(W,A)}] = E_{P_{0,c} \times P_Z}[f_Y(\widetilde{d}_{0,v}(W,A), W, U_Y)],$
and it is identified from $P_0$ by the estimand given in Equation (\ref{equation::fulldataestimand}) under the first part of assumption (A1) and the following additional assumption.
 \begin{enumerate}[resume*=conditions]
\item $U_A \independent  U_Y  \mid W $ $\Big(\text{or } Y \mid \{A=a, W  \} =_d Y_a \mid \{A=a,W\}  \text{ \& } Y_a \independent A \mid W$ for all $a \geq v\Big)$   
\end{enumerate}
Viewing $P(A\geq v \mid W)$ as a threshold-specific propensity score, the first part of assumption (A1) is analogous to the overlap assumption needed for the identification of binary treatment causal effects (Rubin, 2000; Pearl, 2009)\nocite{Rubin2000}. Notably, the overlap condition is fairly mild relative to those for other stochastic interventions on continuous variables (see, e.g. D\'iaz et al., 2021). It only requires that there is some mass assigned to the interval $[v, \infty)$ for all strata of $W$ and no conditions are imposed on how this mass is distributed. Assumption (A3) requires that there are no unmeasured confounders between $A$ and $Y$ and both versions are similar to the assumptions sufficient for identification of other stochastic interventions (D\'iaz et al., 2021)\nocite{diaz2021nonparametric}.

\section{Efficient influence function of the parameter of interest and inefficiency of binary-treatment-based estimators } \label{section::EIF}
A key object necessary for constructing efficient estimators for an estimand is the efficient influence function (EIF), which is uniquely defined given a statistical model and target parameter. The efficient influence function is essential for characterizing the asymptotic distribution of an efficient estimator. Specifically, given the data-generating distribution $P_0$ contained in the nonparametric statistical model $\mathcal{M}$, the parameter $\Psi_v^{adj}$ and its efficient influence function  $D_{P,v}$ indexed by $P \in \mathcal{M}$ and a threshold $v \in \mathbb{R}$, the optimal $\sqrt{n}$-scaled and centered asymptotic distribution among all asymptotically linear and regular (w.r.t. $\mathcal{M}$) estimators of $\Psi_v^{adj}(P_0)$ is mean-zero normally distributed with variance $E_{0} D_{P_0,v}(O)^2$.  For a more detailed account of the theory of efficient influence functions and semiparametric efficiency theory, we refer to Bickel et al. (1993)\nocite{bickel1993efficient}. The derivation of the EIFs can be found in Web Appendix A. 
\begin{lemma}
The efficient influence function of the parameter $\Psi_{v}^{adj}$ is
$$D_{P,v}(W,A, \Delta, \Delta Y) =  \frac{1(A \geq v)}{P(A \geq v  \mid W)}\frac{\Delta }{P(\Delta=1 \mid A,W)}\left(Y -E_P[Y \mid  A, W, \Delta = 1]\right) $$
$$+\left( E_P[Y \mid  A, W, \Delta = 1] -  E_{P_{A\mid W}}[E_P[Y \mid A, W, \Delta = 1] \mid A\geq v, W]\right) \frac{1(A \geq v)}{P(A \geq v  \mid W)}$$
$$+  E_{P_{A\mid W}}[E_P[Y \mid A, W, \Delta = 1] \mid A\geq v, W] - \Psi_{v}^{adj}(P).$$ \label{lemma::EIF}
\end{lemma}
When there is no outcome missingness, i.e. $\Delta \equiv 1$, the EIF reduces to the efficient influence function for the binary treatment-specific mean $E_{P_W} E_{P}[Y \mid D_v=1,\, W]$ where $D_v = 1(A \geq v)$ plays the role of the binary treatment (van der Laan, Robins, 2003; Bang and Robins, 2005; van der Laan, Rose; 2011). For the case of no outcome missingness, this establishes that an estimator for the threshold-response estimand that is nonparametric-efficient w.r.t. the coarsened data-structure $(W, D_v, Y)$ is also efficient with w.r.t. the observed data-structure $(W, A, Y)$. On the other hand, when there is outcome missingness that is only informed by $W$ and not $A$ (so that the estimand reduces to $E_{P_W} E_P[Y| A \geq v, W, \Delta = 1]$), the first two terms of the EIF of Lemma \ref{lemma::EIF} does not reduce to the relevant EIF component for the coarsened data-structure $(W, D_v, \Delta, \Delta Y)$, given by $\frac{\Delta 1(A \geq v)}{P(A \geq v \mid W)P(\Delta = 1\mid W)} \left\{Y - E_P[Y \mid A \geq v,\, W,\, \Delta =1] \right\}$, which implies that estimators based on the coarsened data-structure $O_v = (W, D_v, \Delta, \Delta Y)$ are statistically inefficient with respect to the nonparametric statistical model $\mathcal{M}$. The loss in efficiency of the binary treatment methods is largely driven by how much more $E_P[Y\mid A,\,W,\, \Delta = 1]$ is predictive of $Y$ than $E_P[Y\mid A \geq v,\,W,\, \Delta = 1]$ among individuals with $A \geq v$ (Moore, van der Laan, 2009)\nocite{MooreAdjustment}.

\section{Methodology}
For efficient estimation of the threshold-response function, we employ Targeted Minimum-Loss Based Estimation (TMLE) (van der Laan, Rose, 2011). TMLE is a two-step approach for constructing nonparametric and often double-robust efficient substitution estimators for pathwise-differentiable parameters of interest. TMLE is closely related to the one-step and estimating equation methodology (Bickel et al., 1993; van der Laan, Robins, 2003; Bang and Robins, 2005) in that it utilizes the efficient influence function in its estimation procedure, but it does so in a way that ensures the resulting estimator is a substitution estimator, which can lead to additional robustness and improved performance in certain finite-sample settings (Porter et al., 2011). \nocite{PorterLaanPerform} In short, given an initial estimator $P_{n,0} \equiv (P_{n,W}, Q_n, Q_{n,v}, g_{n,v}, G_n)$ of $P_0 \equiv (P_{0,W}, Q_0, Q_{0,v}, G_{0,v}, G_0)$, the TMLE procedure performs minimum-loss estimation with loss $(P, O) \mapsto L(P, O)$ over a data-adaptive parametric submodel $P_{n, \varepsilon}$ through $P_{n,0}$ that points in the direction of maximal change of the parameter $\Psi_v^{adj}$. We call the resulting risk minimizer $P_{n,v}^*$ the TMLE of $\Psi_v^{adj}$. The submodel and choice of loss function have the key property that the ``likelihood" scores of the risk function along the submodel equals the empirical mean of the efficient influence function of $\Psi_{v}^{adj}$. That is, $\frac{d}{d\varepsilon} P_n L(P_{n,\varepsilon}) = P_n D_{P_{n,\varepsilon}, v}.$ As a consequence, the TMLE, which is the risk minimizer along this submodel, necessarily solves the efficient score equation: $P_n D_{P_{n,v}^*, v} = 0$. This implies $\Psi_v^{adj}(P_{n,v}^*) = \Psi_v^{adj}(P_{n,v}^*) + P_n D_{P_{n,v}^*, v}$, and therefore the substitution estimator using $P_{n,v}^*$ equals the one-step efficient estimator (Bickel et al., 1993) and inherits its asymptotic optimality properties. We refer to Gruber, van der Laan (2009) and van der Laan, Rose (2011) for a more thorough introduction and overview of TMLE. \nocite{TMLEintro}


\subsection{Targeted Minimum Loss Estimator}
We present a sequential-regression estimator (srTMLE) for the threshold-response estimand that applies in the presence of informative outcome missingness. We also provide an inefficient TMLE (binTMLE) which is a consistent estimator in the case of quasi-informative outcome missingness: $A \independent \Delta  \mid W$. The binTMLE is equal to the binary treatment-specific mean TMLE (van der Laan, Rose, 2011\nocite{vanderLaanRose2011}) where the binary treatment is $D_v = 1(A \geq v)$, and is therefore also asymptotically equivalent to the analogous one-step efficient AIPW estimator (see, for instance, Bang and Robins, 2005; van der Laan, Robins, 2003). The binTMLE is statistically inefficient when the missingness is informed by $W$ and inconsistent when informed by $W$ and $A$. In Web Appendix F, we discuss how our proposed estimators can be adjusted when there is missingness (e.g. due to biased sampling) in the biomarker $A$. These methods can also be applied with minor modifications when the outcome $Y$ is bounded and continuous (see Web Appendix I).
 
\subsubsection{A novel efficient sequential-regression-based TMLE} \label{section::seqregTMLE}

The sequential-regression TMLE (srTMLE) requires initial estimation of the nuisance parameters $Q_0(a,w) = E_{0}[Y \mid  A=a,\,W=w,\, \Delta = 1]$, $Q_{0,v}(w) = E_{0, A \mid W}\left[ Q_0(A,W) \mid A \geq v, W=w \right]$, $G_0(a,w) = P_0(\Delta = 1 \mid  A=a,\, W=w)$, and $g_{0,v}(w) = P_0(A \geq v  \mid W=w).$ Denote their respective estimators $Q_n$, $G_n$ and $g_{n,v}$, and let $P_{W,n}$ be the empirical estimator of the marginal distribution of $W$. Our initial estimator of (the relevant parts of) $P_0$ is given by $P_{n,v,0} := (P_{W,n}, g_{n,v}, Q_{n}, G_n)$.
The srTMLE is defined as follows.
\begin{enumerate}
    \item Define the indicator fluctuation submodel
$Q_{n, \varepsilon} (A,W)= \text{expit} \left \{\text{logit}(Q_{n})(A,W) + \varepsilon 1(A \geq v)  \right\}. $ 
    \item The MLE along this submodel is given by
    $$\hat \varepsilon_n = \argmax_{\varepsilon \in \mathbb{R}} \frac{1}{n} \sum_{i=1}^n  \frac{\Delta_i }{g_{n,v}(W_i) G_n(A_i,W_i)} \left \{ Y_i \cdot \log Q_{n, \varepsilon}(A_i, W_i) + (1- Y_i) \cdot \log (1 - Q_{n,  \varepsilon}(A_i, W_i)) \right\}.$$
    \item Define the updated estimate of $Q_0$ as $Q_{n}^* = Q_{n, \hat \varepsilon_n}$.
    \item Obtain an initial estimator $Q_{n,v}$ of $Q_{0,v}(W) = E_{P_0}[E_{P_0}[Y \mid A, W, \Delta = 1] \mid  A \geq v, W] $ using sequential regression (e.g. estimate $E[Q_n^*(A,W) \mid  A \geq v, W]$).
    \item Define the intercept fluctuation submodel,
$Q_{n, v, \varepsilon}= \text{expit} \left \{\text{logit}(Q_{n,v}) + \varepsilon  \right\}. $
    \item The MLE along this submodel is given by $Q_{n,v}^* = Q_{n,v, \hat \varepsilon_n}$ where
    $\hat \varepsilon_n =$
$$\argmax_{\varepsilon \in \mathbb{R}} \frac{1}{n} \sum_{i=1}^n  \frac{1(A_i \geq v)}{g_{n,v}(W_i)} \left \{ Q_{n}^*(A_i,W_i) \log Q_{n, v, \varepsilon}(W_i) + (1- Q_n^*(A_i,W_i)) \log (1 - Q_{n,  v, \varepsilon}(W_i)) \right\}.$$
    \item Let $P_{n,v}^* := (P_{W,n}, g_{n,v}, G_n, Q_{n,v}^*, Q_n^*)$. The TMLE of $\Psi_{v}^{adj}(P_0)$ is then given by the substitution estimator
$\Psi_{v}^{adj}(P_{n,v}^*) = E_{P_{W,n}}  Q_{n,v}^*(W)= \frac{1}{n} \sum_{i=1}^n Q_{n,v}^*(W_i).$
\end{enumerate}
Step (2) (resp. (6)) requires performing the IPW-weighted logistic regression of $Y$ (resp. $Q_n(A,W)$) on $1(A \geq v)$ (resp. (1)) with offset being an initial estimator of $Q_n(A,W)$ (resp. $Q_{n,v}(W)$).
In step (5), the initial estimator of $Q_{0,v}$ can be obtained by treating $Q_n^*(A,W)$ as a pseudo-outcome and then performing the nonparametric regression of $Q_n^*(A,W)$ on $W$ using only the observations $O_i$ with $A_i \geq v$. The key property of the srTMLE is that the targeted estimators $Q_{n}^*$ and $Q_{n,v}^*$ solve the following score equations: $\frac{1}{n} \sum_{i=1}^n \frac{1(A_i \geq v, \Delta_i=1)}{g_{n,v}(W_i) G_n(A_i, W_i)} (Y_i - Q_{n}^*(A_i, W_i)) = 0$
and
$\frac{1}{n} \sum_{i=1}^n \frac{1(A_i \geq v)}{g_{n,v}(W_i) } (Q_{n}^*(A_i, W_i) - Q_{n,v}^*(W_i)) = 0,$
which implies that the efficient score equation is solved: $P_n{D}_{P_{n,v}^*,v} = 0.$

\subsubsection{Inefficient TMLE for quasi-informative missingness}
\label{section::binaryTMLE}
 The inefficient binary-treatment-based TMLE (binTMLE) is consistent, but inefficient, if  $A \independent \Delta \mid W, A \geq v$. We use the same notation as in the previous section except let $G_{n,v}(w)$ be an estimator of $P_0(\Delta = 1 \mid A \geq v, \, W =w)$ and $Q_{n,v}(w)$ be an estimator of $E_0[Y \mid A \geq v, \, W , \, \Delta = 1]$ (which equals $Q_{0,v}$ under the assumption that $A \independent \Delta \mid W, A \geq v$).
\begin{enumerate}
    \item Define the intercept fluctuation submodel
    $Q_{n, v, \varepsilon} = \text{expit} \left \{\text{logit}(Q_{n,v}) + \varepsilon  \right\}. $
    \item Define the MLE along this submodel $Q_{n,v}^* = Q_{n,v, \hat \varepsilon_n}$ by
$$\hat \varepsilon_n = \argmax_{\varepsilon \in \mathbb{R}} \frac{1}{n} \sum_{i=1}^n \frac{\Delta_i 1(A_i \geq v)}{G_{n,v}(W_i) g_{n,v}(W_i)} \left \{ Y_i \cdot \log Q_{n, v, \varepsilon}(W_i) + (1- Y_i) \cdot \log (1 - Q_{n, v, \varepsilon}(W_i)) \right\}.$$
    \item Let $P_{n,v}^* = (P_{W,n}, g_{n,v}, G_{n,v}, Q_{n,v}^*)$. The binTMLE for $\Psi_v^{adj}(P_0)$ is given by the substitution estimator
    $\Psi_v^{adj}(P_{n,v}^*) = \frac{1}{n} \sum_{i=1}^n Q_{n,v}^*(W_i).$
\end{enumerate}

\subsection{Asymptotic inference with the srTMLE}

In this section, we present a general theorem that characterizes the asymptotic behavior of the srTMLE defined in Section \ref{section::seqregTMLE}. An analogous result for the binTMLE defined in section \ref{section::binaryTMLE} follows from van der Laan, Rose (2011). Let $P_{n,v}^* = (P_{W,n}, g_{n,v}, G_n, Q_n,  Q_{n,v}^*)$ denote the targeted nuisance estimates for the srTMLE of $\Psi^{adj}_v(P_0)$. Let $K \subset \mathbb{R}$ be a compact set. In order for the srTMLE to be asymptotically linear and efficient, we require the following regularity conditions on the initial estimators and nuisance functions.  
\begin{enumerate}[series=conditions, label=({B}{{\arabic*}})]
    \item  $\delta < P_0(Y = 1 \mid A \geq v) < 1- \delta$ for some $\delta > 0$.
    \item $g_{0,v}, G_{0} > \delta$ and $g_{n,v}, G_n > \delta$ with probability tending to 1 for some $\delta >0$.
    \item  The set of realizations of $w \mapsto g_{n,v}(w)$, $(a,w) \mapsto G_n(a,w)$, $(a,w) \mapsto Q_n(a,w)$, and $w \mapsto Q_{n,v}(w)$ are $P_0$-Donsker.  
    \item  $\norm{Q_{n} - Q_{0}} = o_P(n^{-1/4})$, $\norm{Q_{n,v} - Q_{0,v}} = o_P(n^{-1/4})$, $\norm{g_{n,v} - g_{0,v}}  = o_P(n^{-1/4})$, $\norm{G_{n} - G_{0}}  = o_P(n^{-1/4})$.
    \item  The union over all $v \in K$ of the set of realizations of the functions given in B3 is $P_0$-Donsker. 
\end{enumerate}

Condition B1 requires that the random outcome $Y$ is not degenerate on the event $\{A \geq v\}$, which ensures that the efficient influence function is nonvanishing. This condition is required for the srTMLE to have an asymptotic distribution after centering and scaling by $\sqrt{n}$ but is not required for $\sqrt{n}$-consistency of the srTMLE. Condition B2 is a standard positivity assumption needed for the estimand of interest and srTMLE procedure to be well-defined. Condition B3 requires that the estimators of the nuisance parameters are well-behaved with their realizations falling in a not-too-complex function space. More aggressive algorithms like 
random forests and gradient boosting are prone to overfitting when not properly tuned, which can lead to a violation of this condition. However, by cross-fitting the nuisance estimators and employing CV-TMLE (van der Laan, Robins, 2011; Chernozhukov et al., 2016)\nocite{DoubleML}, this condition can be removed entirely, allowing one to safely employ such estimators (see references for more discussion). Condition B4 requires the nuisance estimators converge fast enough to the nuisance parameters. Both condition B3 and B4 are satisfied under smoothness conditions by a number of estimators including risk minimizers over reproducing Kernel Hilbert spaces (RKHS), neural networks with VC dimension that does not grow too fast with sample size (Farrell et al., 2018), generalized additive models, and the Highly Adaptive Lasso estimator (Benkeser and van der Laan, 2016). Condition B5 ensures that the nuisance estimator realizations and nuisance parameters for all thresholds $v \in K$ fall in a single controlled function class. This condition is trivially satisfied for the nuisance estimators whenever condition B3 is satisfied and $K$ is finite. Otherwise, this can usually be enforced by pooling the estimation across the thresholds and then employing a machine-learning algorithm with $P_0$-Donsker realizations (e.g. use pooled logistic/linear regression). Alternatively, this can be enforced by ensuring that all nuisance estimators for each $v \in K$ fall in a single Donsker function class (e.g. functions of bounded variation or a reproducing kernel Hilbert space with uniformly bounded Hilbert space norm). This can be guaranteed by estimating the nuisance functions with the Highly Adaptive Lasso estimator (Benkeser, van der Laan, 2016) as long as the nuisance functions have a uniformly bounded variation norm. 

\begin{theorem}
Suppose conditions B1, B2, B3 and B4 hold. Then, the srTMLE estimator $\Psi_v^{adj}(P_{n,v}^*)$ satisfies
$$\sqrt{n} (\Psi_v^{adj}(P_{n,v}^*) - \Psi_v^{adj}(P_0)) = n^{-1/2} \sum_{i=1}^n D_{P_0,v}(W_i, A_i, \Delta_i, \Delta_i Y_i) + o_P(1).$$
If only B1 is violated then one has $(\Psi_v^{adj}(P_{n,v}^*) - \Psi_v^{adj}(P_0)) = o_P(n^{-1/2})$. If in addition assumption B4 holds uniformly for all $v$ in the bounded set $K \subset \mathbb{R}$ and assumption B5 holds, then $\left( \sqrt{n} (\Psi_v^{adj}(P_{n,v}^*) - \Psi_v^{adj}(P_0)): v \in K\right)$ converges to a tight mean-zero Gaussian process in $l^\infty(K)$ with covariance function $\rho(v_1,v_2) = P_0 D_{P_0,v_1} D_{P_0,v_2}.$

\end{theorem}

It follows immediately from the preceding theorem that the srTMLE is an efficient estimator for $\Psi_v^{adj}(P_0)$, since it is asymptotically linear with influence function being the efficient influence function. The srTMLE's scaled and centered asymptotic distribution is given by a mean-zero normally-distributed random variable with variance being the variance of its influence function. An estimate $\sigma_{n,v}$ of the standard error $\sigma_{v} := \sqrt{P_0 D_{P_0,v}^2}$ of the srTMLE is given by $\sigma^2_{n,v} = \frac{1}{n} \sum_{i=1}^n D_{P_{n,v}^*, v}(W_i, A_i, Y_i)^2$. Under condition B3, one has that $D_{P_{n,v}^*, v}$ falls in a class of $P_0$-Donsker functions, which when paired with condition B4 and the fact that $D_{P_0,v}$ is bounded under condition B2, implies that $ \mid \sigma_v - \sigma_{n,v} \mid  = o_P(1).$ Using this estimate of the standard error, Wald-type confidence intervals can be constructed for inference. Specifically, $\Psi_v^{adj}(P_{n,v}^*) \pm \Phi^{-1}(1-\alpha) \frac{\sigma_{n,v}}{\sqrt{n}}$ forms an asymptotic $1-\alpha$ confidence interval for $\Psi_v^{adj}(P_0)$ where $\Phi$ is the CDF of a $N(0,1)$ random variable.

\subsection{Robustness properties of the srTMLE estimator}
The srTMLE is double robust with respect to the nuisance estimators. The double-robustness property is due to the structure of the efficient influence function and is analogous to the double-robustness property for the well-known AIPW estimator of the average treatment effect (Bang and Robins, 2005)\nocite{bangrobins}. 

\begin{theorem}
Under conditions B2, B3, the srTMLE given in section 5.1.1. is a consistent estimator for $\Psi_{v}^{adj}(P_0)$ if either of the following conditions hold. If condition B4 holds as well then the srTMLE is consistent uniformly in $v \in K$. 
\begin{itemize}
    \item $P_0(A \geq v \mid W)$ and $P_0(\Delta = 1 \mid A, W)$ are estimated consistently.
    \item $E_{0, A \mid W}[E_{0}[Y \mid A, W, \Delta = 1] \mid A\geq v,W]$ and $E_{0}[Y \mid A, W, \Delta = 1]$  are estimated consistently.
\end{itemize}
\end{theorem}

\subsection{Simultaneous confidence bands for the threshold-response function}
Let $V = \{v_1, \dots, v_k\} \subset K \subset \mathbb{R}$ be a finite set of thresholds for some $v \in \mathbb{N}$, contained in a bounded set $K$. In practice, $V$ represents a discrete grid of threshold values that are of interest. Note that since $V$ is discrete, we can apply Theorem 1 with $K = V$. Thus, we have that the collection of srTMLEs $\lbrace \Psi_v^{adj}(P_{n,v}^*): v \in V \rbrace $ satisfies 
$\left\{\sqrt{n} \left(\Psi_v^{adj}(P_{n,v}^*) - \Psi_v^{adj}(P_0) \right) : v \in V\right\} \rightarrow \left\{Z_v: v \in V \right\},$
where $(Z_v: v \in V)$ is a mean-zero multivariate normally distributed random variable with covariance matrix $\Sigma_{v_1, v_2} = P_0 D_{P_0,v_1} D_{P_0,v_2}.$ The covariance matrix can be estimated consistently with the empirical covariance matrix $\hat \Sigma_{n, v_1, v_2} = P_n D_{P_{n,v}^*,v_1} D_{P_{n,v}^*,v_2}$, which further gives a consistent estimate of the distribution of $(Z_v: v \in V)$. Simultaneous confidence intervals for multivariate normally distributed random variables are well understood and we refer to Cai and M. van der Laan (2019)\nocite{SurvOneStepTmle} for an in-depth-treatment in the context of TMLE. By the second statement of Theorem 1, the inference remains asymptotically valid even as we let the number of thresholds contained in $V$ grow and approach the set $K$. Thus, by taking a sufficiently fine grid $V$ of thresholds, we can interpolate the estimates and confidence bands of the thresholds in $V$ to thresholds in $K \cap V^c$ at negligible cost in bias and coverage. We note that since $v \mapsto \Psi_v^{adj}(P_{n,v}^*)$ is a locally efficient, asymptotically linear and regular estimator, these confidence intervals are non-adaptive. In fact, by the local asymptotic minimax theorem (van der Vaart, 1998)\nocite{vaart_1998}, any adaptive estimator of $v \mapsto \Psi_v^{adj}(P_{0})$ that achieves better performance than the efficient estimator at some distribution that is a $n^{-1/2}$-fluctuation from $P_0$ must necessarily be suboptimal at some distributions, and therefore is not locally efficient. This is in contrast with constructing confidence intervals for general regression functions for which no such efficiency theory exists (Genovese, Wasserman, 2008)\nocite{AdaptCI}.

\section{Simulations}
\subsection{Asymptotic efficiency gains of srTMLE relative to binTMLE}
In this section, we explore the efficiency claims made in Section \ref{section::EIF}. Specifically, we investigate the efficiency gains from using the srTMLE as opposed to the binTMLE. For \text{const} $\in \{0,1\}$ and \text{offset} $\in \{0, -3\}$, we consider the following simulation setting:
$A \sim \text{truncnorm}(a = 0, b= 2, \text{mean} =  (0.8 + W_1 +  (W_2 + W_3)/2 )/2, \text{sd} = 0.5)$, 
$Y \sim \text{Bern}( \text{expit}( \text{offset} + 0.75\cdot W_1 -0.2 +  0.5\cdot (W_2 + W_3) - A + \text{const} \cdot 2 \cdot \sin(6 \cdot A) ))$,
$\Delta \sim \text{Bern}(\text{expit}(   -1 + W_1 +   W_2 + W_3 )) \text{ (on average 40\% missing)}$ where $W_1,W_2,W_3$ are baseline variables (See Web Appendix D).  For const $=1$, the addition of the term $2 \cdot \sin(6 \cdot A)$ in the distribution of $Y$ ensures there is a non-linear association between $A$ and $Y$ that will not be captured well by $v \mapsto E_0[Y| A \geq v,\, W]$. \text{offset} $=0$ corresponds with a non-rare event setting with $P_0(Y=1) \approx 0.45$, and \text{offset} $ =-3$ corresponds with a rare-event setting with  $P_0(Y=1) \approx 0.08$. For a grid of 10 thresholds in $[0,2]$, we estimated (using a large simulated sample) the relative loss in efficiency, defined as $ \frac{sd(D_{P_0,v, coarse})}{sd(D_{P_0,v})} - 1$ where $D_{P_0,v, coarse}$ is the inefficient influence function based on the coarsened data-structure $(W, D_v, \Delta, \Delta Y)$. $sd(D_{P_0,v, coarse})$ is the asymptotic standard error of the binTMLE and $sd(D_{P_0,v})$ is the asymptotic standard error of the srTMLE. The results for all settings is displayed in Figure \ref{Figure::effGains}. Both the rare and non-rare event settings show that there is a noticeable efficiency loss for the case where $E_0[Y \mid A,\, W]$ is non-linear. This makes sense because $E_0[Y \mid A,\, W]$ is more predictive of $Y$ $(r = 0.60, r_{rare} = 0.35)$ than $E_0[Y \mid A \geq v,\, W]$ ($r=0.08, r_{rare} = 0.07$ for $v=0$) (see the discussion in Section \ref{section::EIF}). The loss in efficiency is much more substantial in the non-rare-event setting, which can be explained by the fact that in the rare-event-setting both $E_0[Y \mid A,\, W]$ and $E_0[Y \mid A \geq v,\,W]$ are usually small and are therefore poor predictors of $Y$. Next, we see that in both linear settings, there is little-to-no loss in efficiency. This can be explained by (1) the decrease in the predictive power of $E_0[Y \mid A,\, W]$ that is directly due to $A$ by omitting the non-linear term $(r = 0.24,\, r_{rare} = 0.1)$ and (2) the monotone relationship between $E_0[Y \mid A,\, W]$ and $A$ that allows for $E_0[Y \mid A \geq v,\, W]$ to be more predictive ($r = 0.15,\, r_{rare} = 0.07$ for $v=0$) than it was in the non-linear case.

\begin{figure}[ht]
\centering 
\includegraphics[width=15cm,height = 9cm]{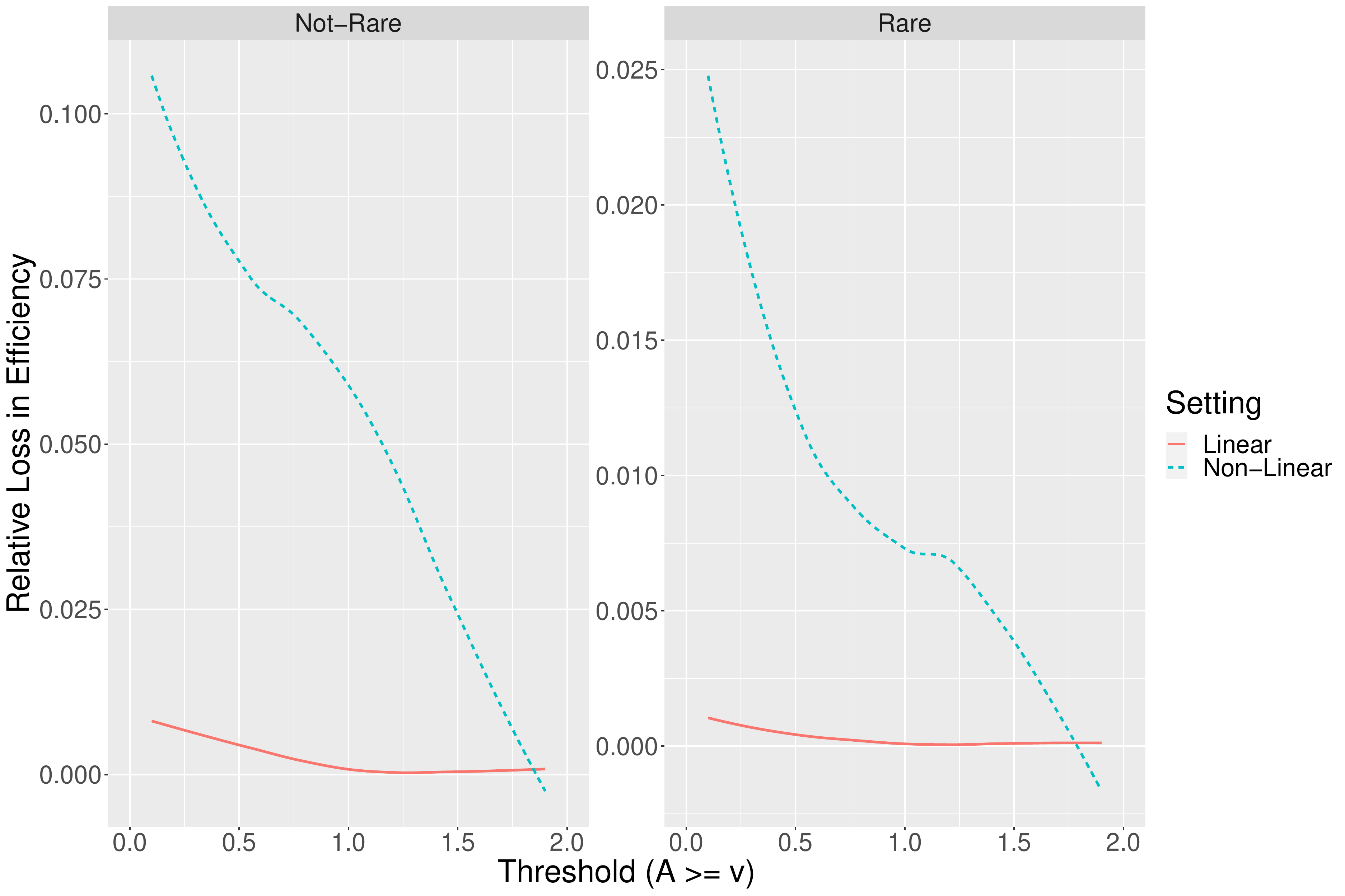}
\caption{Plot of efficiency loss of binary treatment estimator relative to proposed efficient estimator of threshold-response as a function of the threshold for the following simulation settings (in order of curves from top to bottom): Non-linear $Q$, Non-linear $Q$ with rare-events, Linear $Q$, Linear $Q$ with rare-events. This figure appears in color in the electronic version of this article, and any mention of color refers to that version.
}
\label{Figure::effGains}
\end{figure}

 \subsection{Asymptotic bias of inefficient (coarsened data-structure) TMLE}
In this section, we investigate how the following estimators perform in a setting with outcome missingness that is informed by both $W$ and $A$: (1) The proposed efficient srTMLE; (2) The inefficient binTMLE; (3) The estimator of Donovan et al. All nuisance functions were estimated using generalized additive models using the function ``gam" provided in the R package ``mgcv"  with default settings. For the distribution of $A$ and $Y$, we utilize the same distribution as the previous section with $\text{offset} = -3$ and $\text{const} = 1$, which corresponds with rare-events with non-linear $Q$. For $\text{const}_2 \in \{0,1\}$, we generate the missingness indicator in two different ways as follows:
$\Delta \sim \text{Bern}\left(\text{expit}(  -1 +  A + \text{const}_2 \cdot 2 \cdot \sin(6\cdot A) + W_1 + (W_2 +W_3)/2 )  \right)$ (on average 35\% missing).
$\text{const} = 1$ corresponds with complex outcome missingness and $\text{const}_2 = 0$ corresponds with simple outcome missingness. Due to confounding bias, we expect the estimator of Donovan et al. to perform poorly. We also expect the binTMLE to be asymptotically biased due to the strong dependence of $\Delta$ on $A$. The results of the simulations are given in Figure \ref{Figure::SimBias} and Figure \ref{Figure::SimBiasLess}. For the complex outcome missingness case, both the Donovan et al. estimator and the binTMLE are biased with especially poor confidence interval coverage as sample size increases. The srTMLE obtains the nominal 95\% confidence interval coverage at around sample size $n=1000$. For very small sample sizes, the bias of the binTMLE is comparable to that of the srTMLE, which is likely due to the sinusoidal signals being indistinguishable from noise and the resulting finite sample bias happening to be favorable for the binTMLE. 
For the simple outcome missingness case, given in Figure \ref{Figure::SimBiasLess}, we see both the binTMLE and srTMLE perform similarly in coverage and standard error. Even though the outcome missingness is informed by the biomarker, the binTMLE is only slightly biased and therefore still performs well. This suggests in some cases that the binTMLE, although biased, can still perform well when the outcome missingness is informed by the biomarker.

 \begin{figure}[ht]
\begin{subfigure}{.5\textwidth}
  \centering
   \caption{Standard error of estimator}
  \includegraphics[width=.8\linewidth]{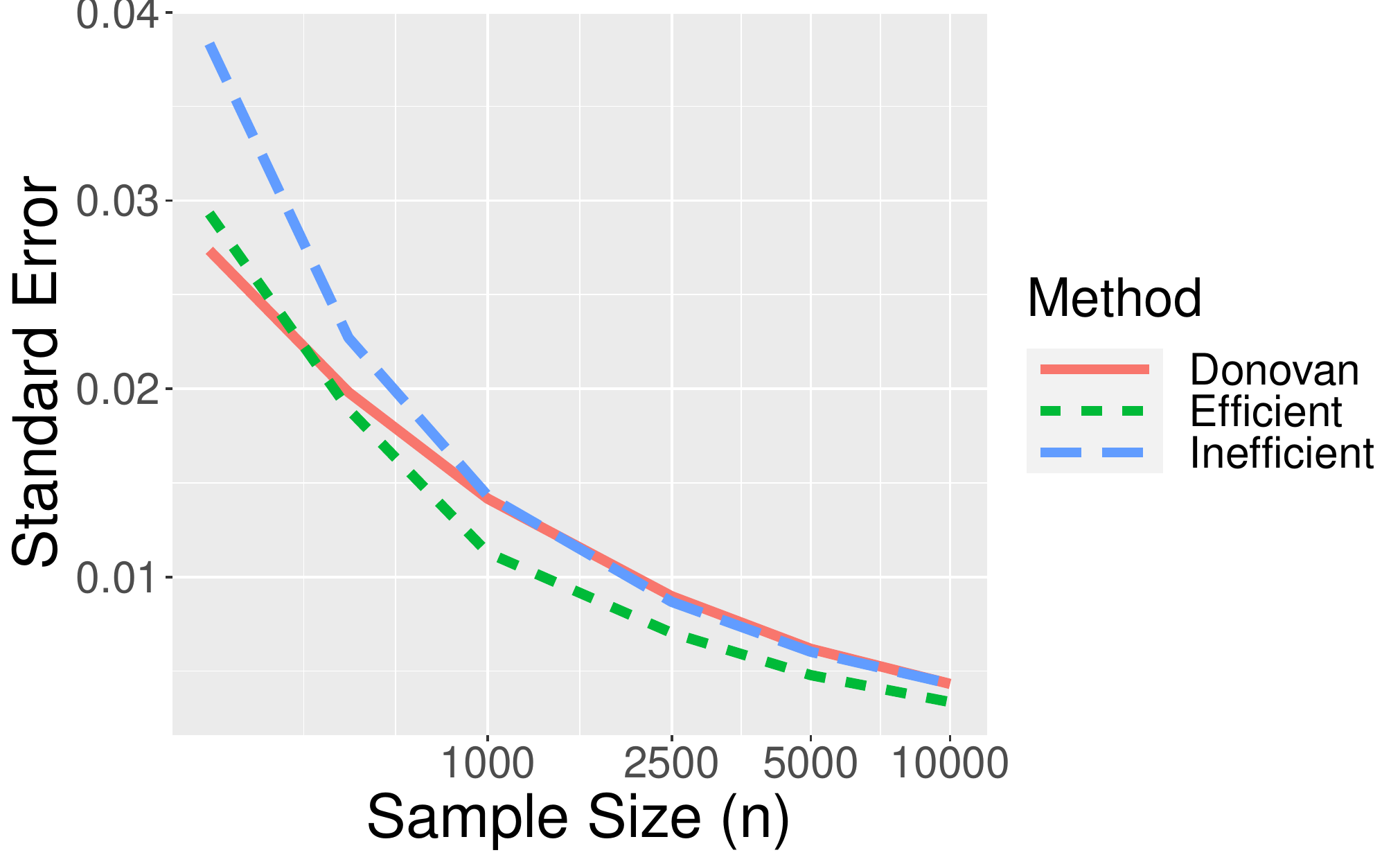}

\end{subfigure}
\begin{subfigure}{.5\textwidth}
  \centering
  
  \includegraphics[width=.8\linewidth]{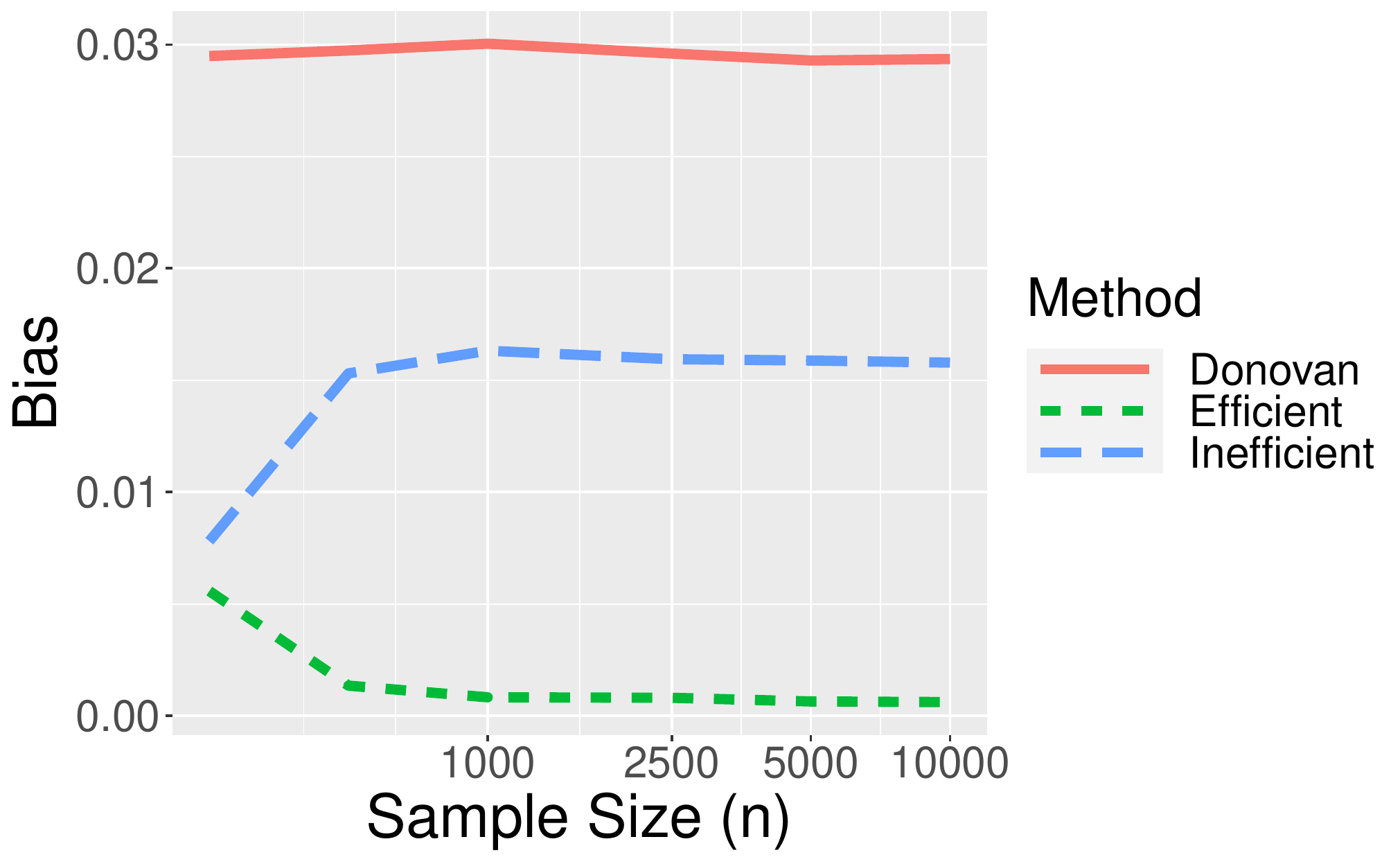}
   \caption{Absolute bias of estimators}

\end{subfigure} 
\begin{subfigure}{.5\textwidth}
  \centering
    \caption{$\sqrt{\text{Mean-squared-error}}$}
  \includegraphics[width=.8\linewidth]{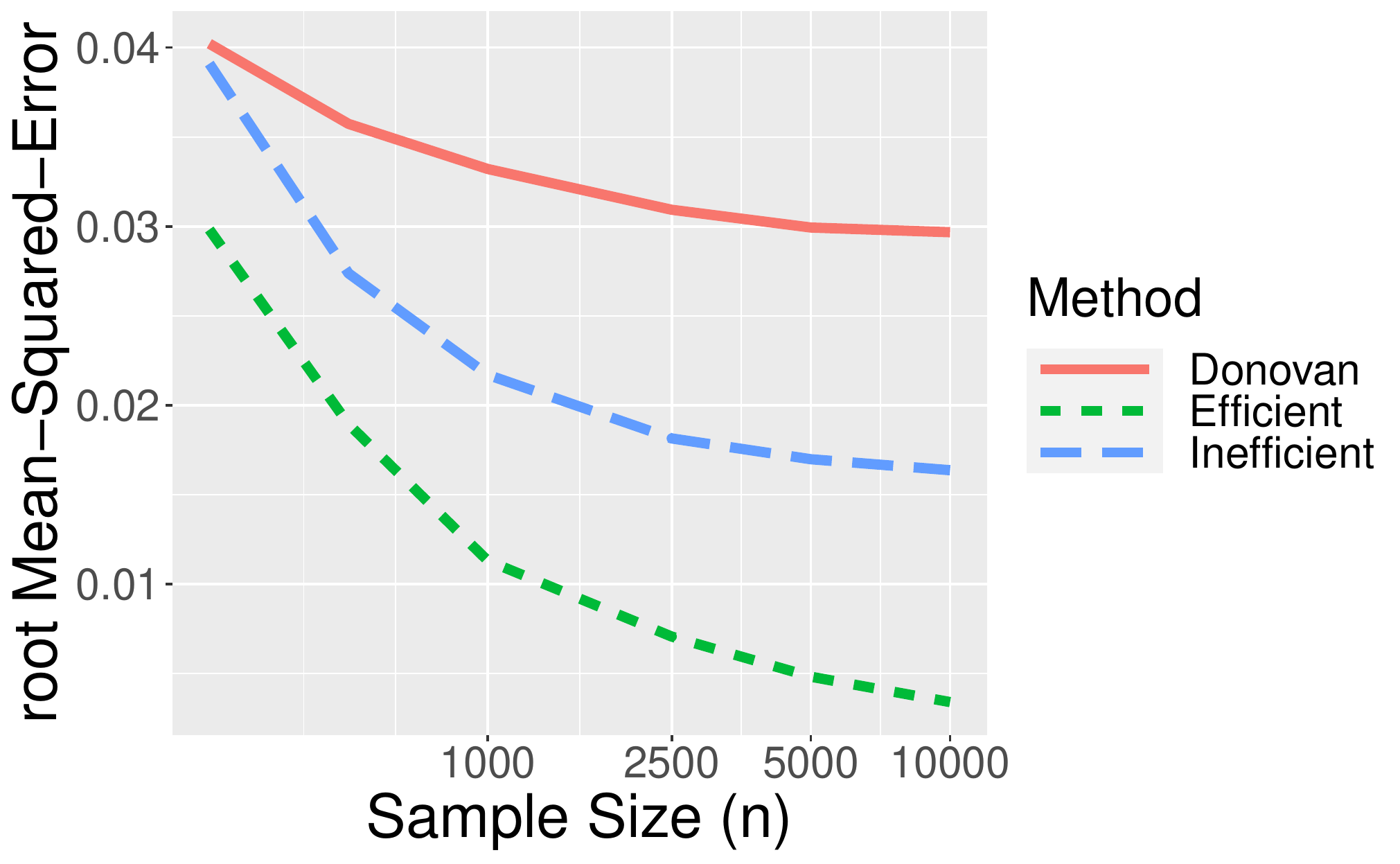}

\end{subfigure}
\begin{subfigure}{.5\textwidth}
  \centering
  
  \includegraphics[width=.8\linewidth]{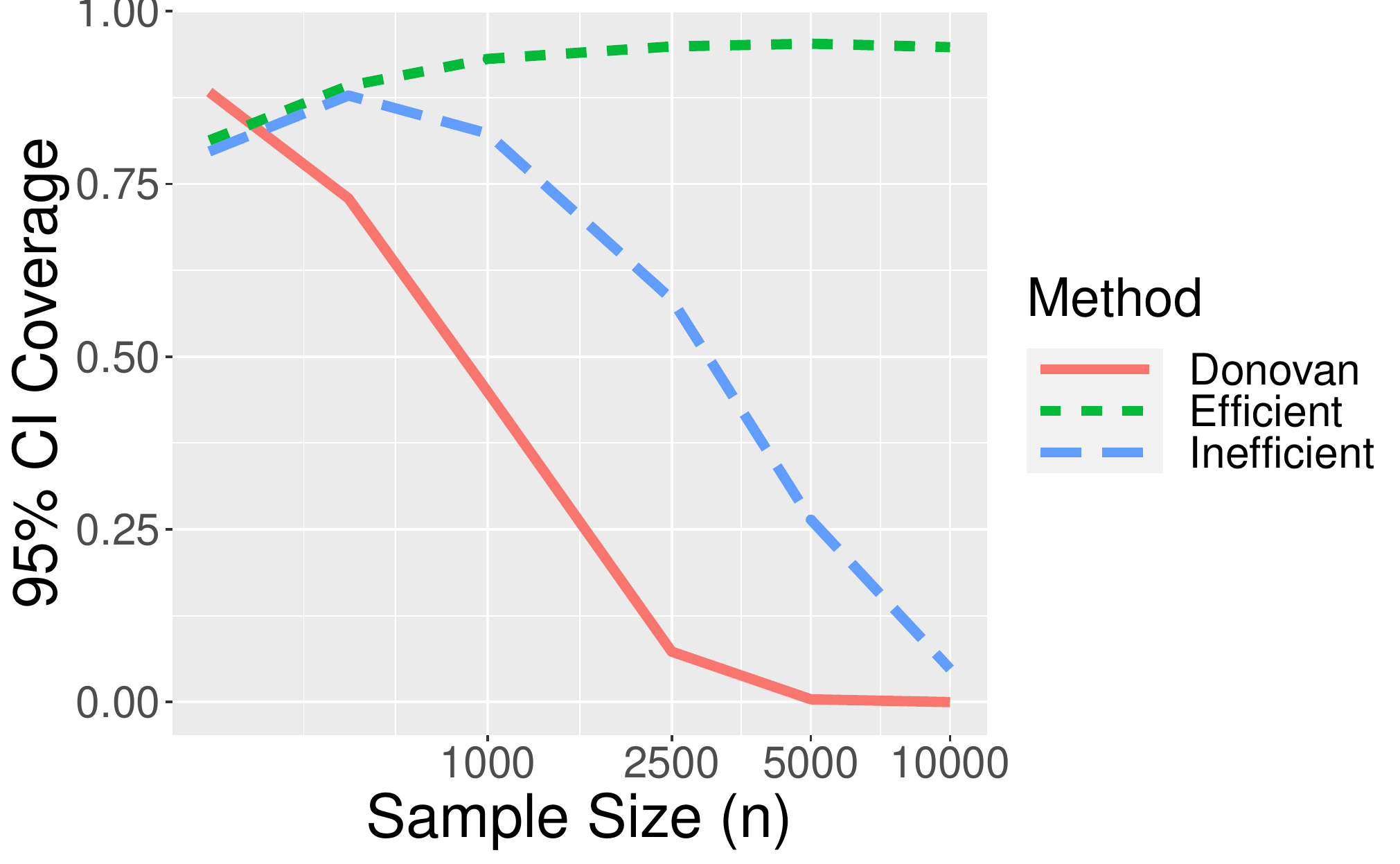}
   \caption{Coverage probability of estimated $95\%$ CI}

\end{subfigure} 
\caption{Results for second simulation (complex outcome missingness): (a) Standard error, (b) Absolute bias, (c) mean-squared error, and (d) confidence interval coverage computed from 1000 monte-carlo simulations for the proposed srTMLE, the binTMLE, and the Donovan unadjusted estimator. This figure appears in color in the electronic version of this article, and any mention of color refers to that version.
 }
\label{Figure::SimBias}
\end{figure}

 \begin{figure}[ht]
\begin{subfigure}{.5\textwidth}
  \centering
   \caption{Standard error of estimator}
  \includegraphics[width=.8\linewidth]{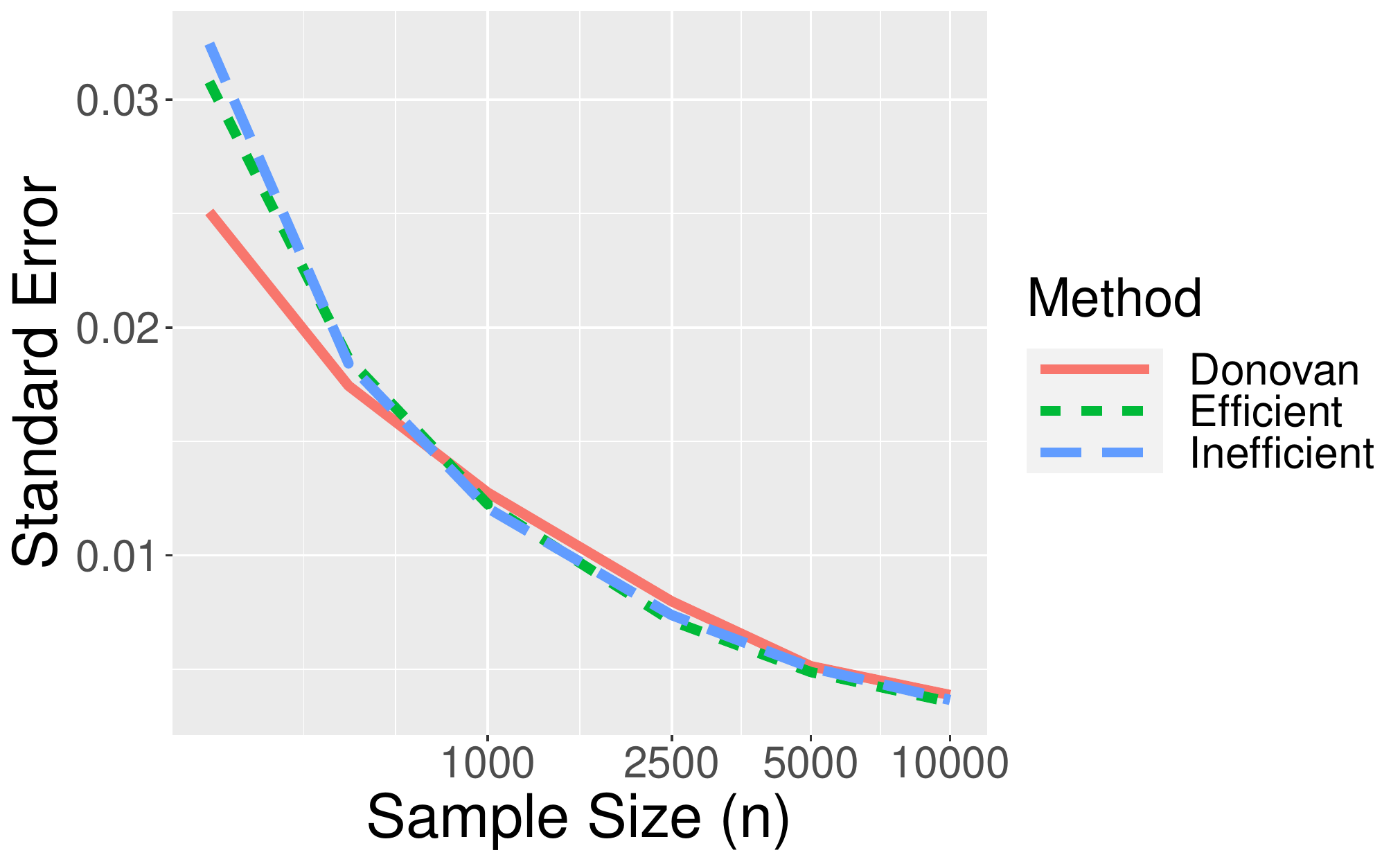}

\end{subfigure}
\begin{subfigure}{.5\textwidth}
  \centering
  
  \includegraphics[width=.8\linewidth]{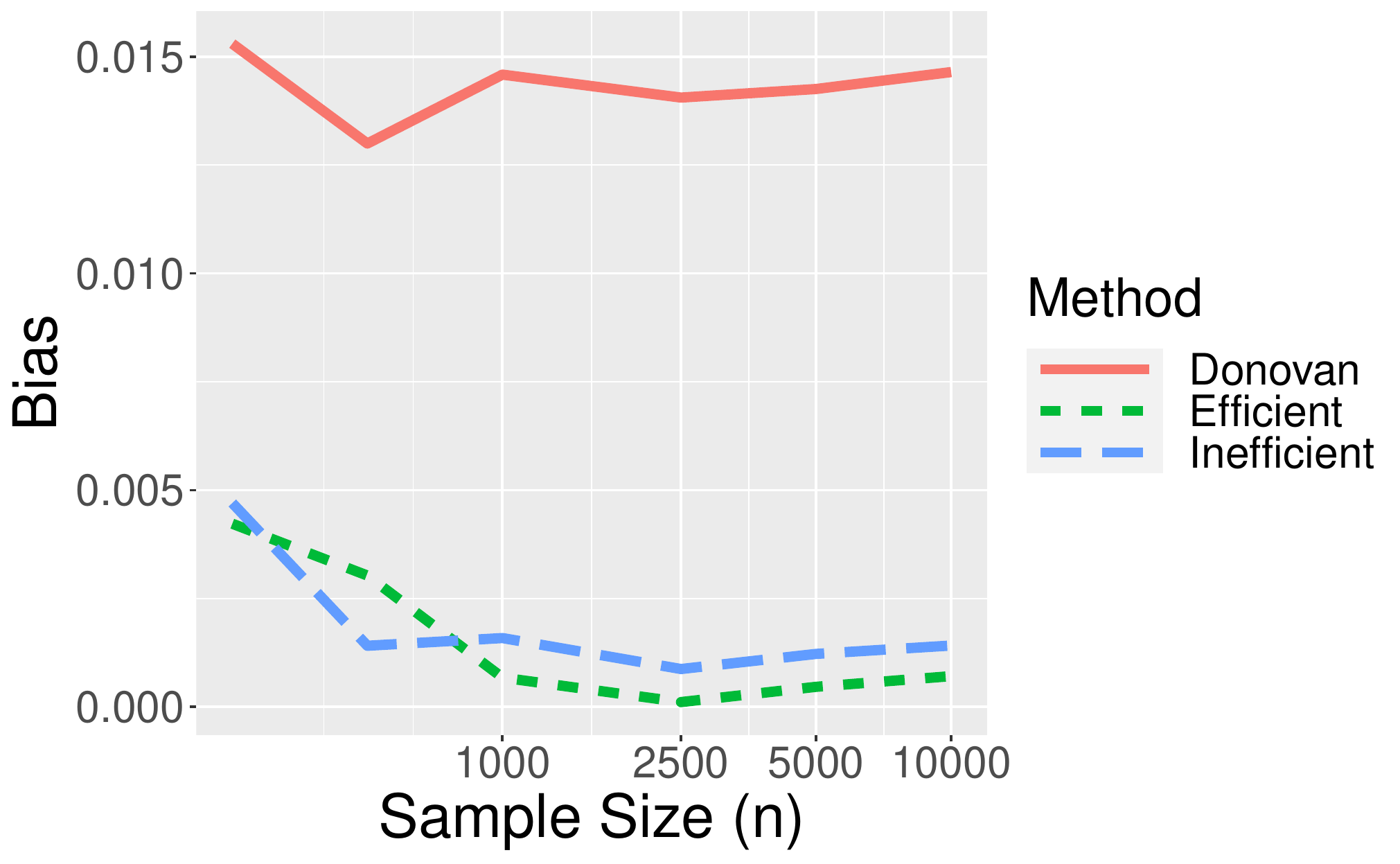}
   \caption{Absolute bias of estimators}

\end{subfigure} 
\begin{subfigure}{.5\textwidth}
  \centering
    \caption{$\sqrt{\text{Mean-squared-error}}$}
  \includegraphics[width=.8\linewidth]{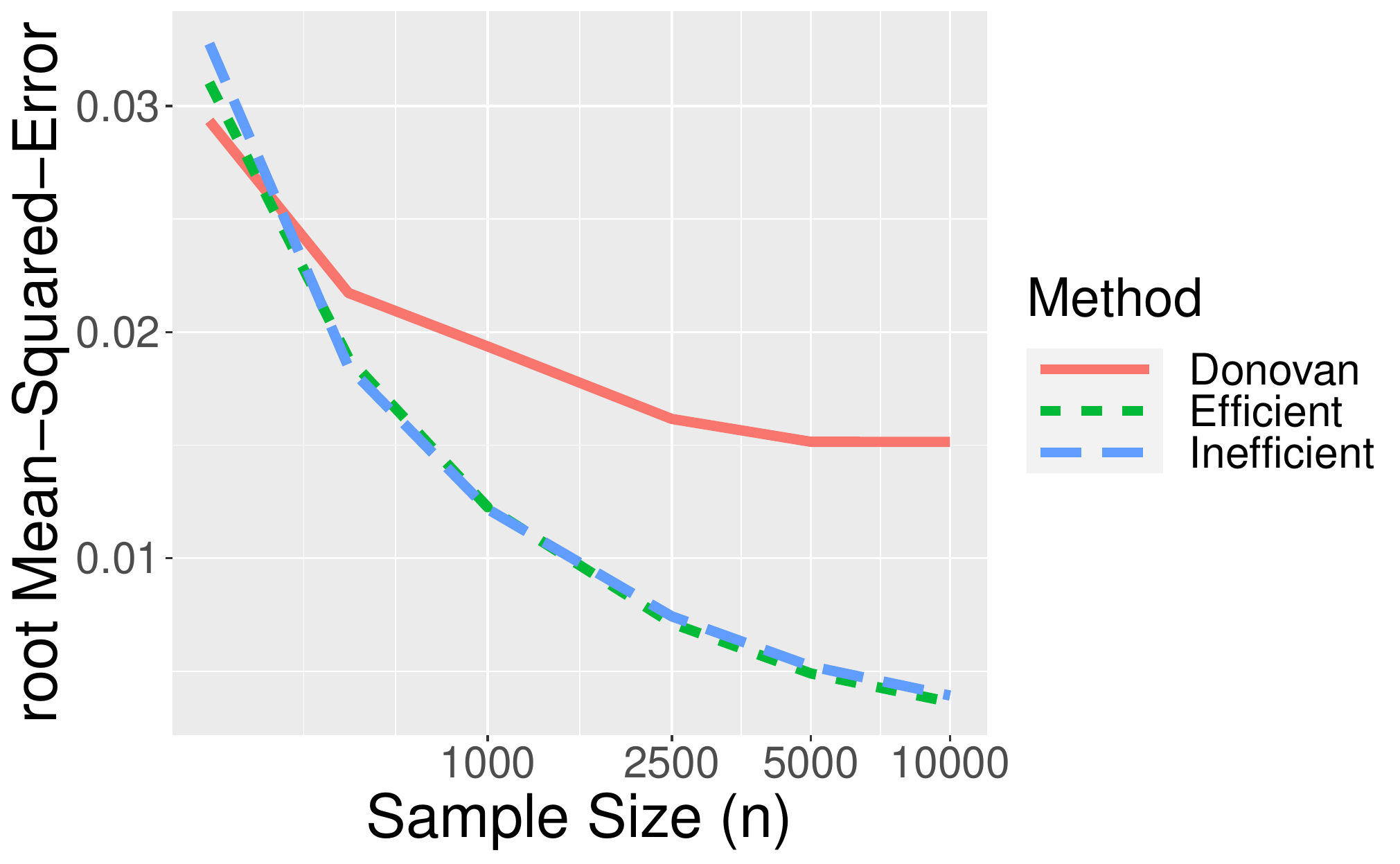}

\end{subfigure}
\begin{subfigure}{.5\textwidth}
  \centering
  
  \includegraphics[width=.8\linewidth]{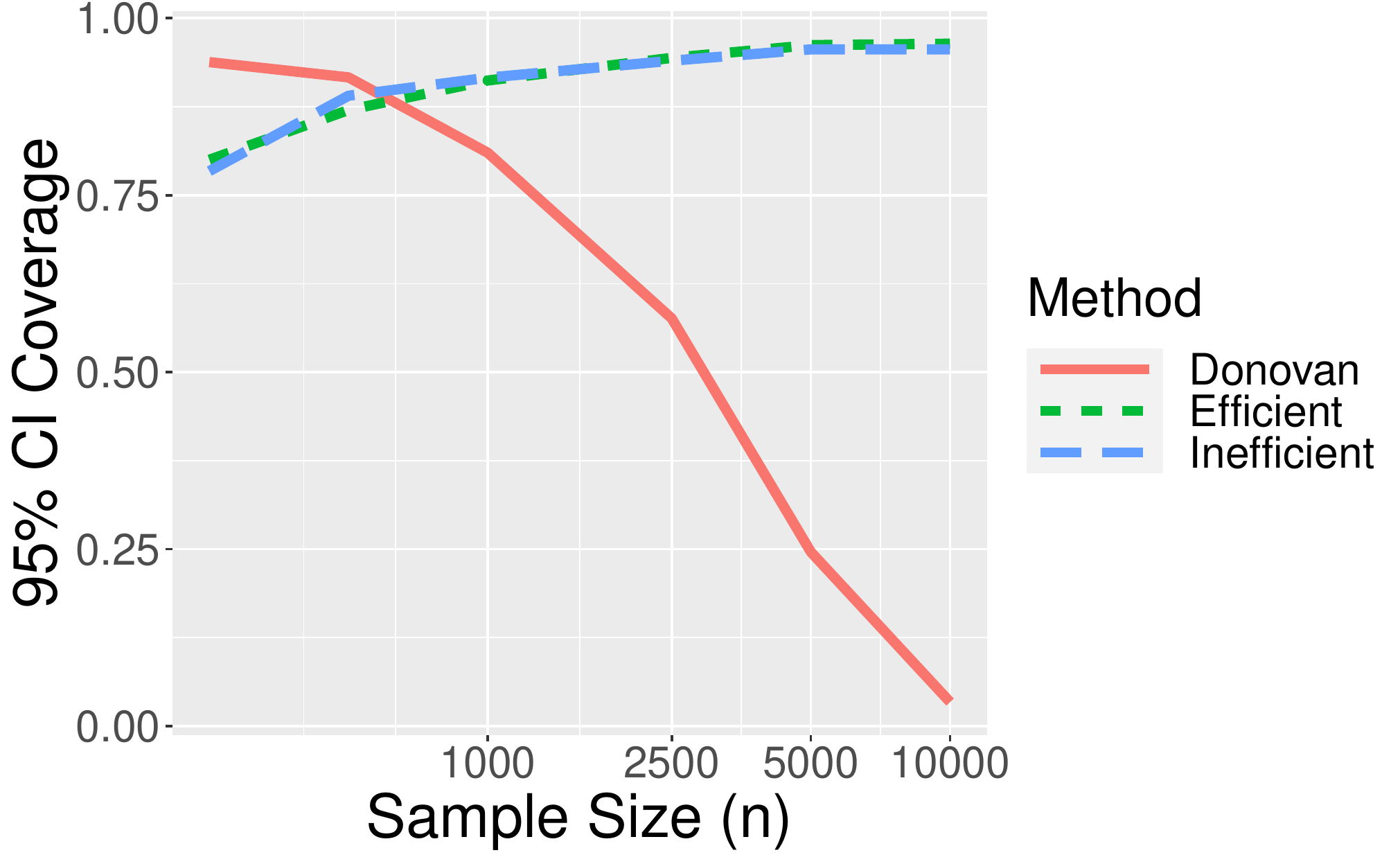}
   \caption{Coverage probability of estimated $95\%$ CI}

\end{subfigure} 
\caption{Results for second simulation (simple outcome missingness): (a) Standard error, (b) Absolute bias, (c) mean-squared error, and (d) confidence interval coverage computed from 1000 monte-carlo simulations for the proposed srTMLE, the inefficient binTMLE, and the Donovan unadjusted estimator. This figure appears in color in the electronic version of this article, and any mention of color refers to that version. }
\label{Figure::SimBiasLess}
\end{figure}

\subsection{Comparison of unadjusted Donovan estimator vs adjusted srTMLE at various levels of confounding}
To understand the effect of covariate adjustment in reducing confounding bias, we evaluate the absolute bias $\left | E_{0}\left[\Psi_v^{adj}(P_{n,v}^*)-\Psi_v^{adj}(P_0)\right] \right | $ of the Donovan et al. (unadjusted) estimator and (covariate-adjusted) srTMLE at sample sizes $n=500,1000,2000$ for rare event settings with varying levels of confounding. Holding the correlation between $W$ and $A$ fixed at $0.5$, we measured the degree of confounding as the amount of correlation between the univariate confounding variable ($W$) and the outcome ($Y$). Across all simulations, we kept the average risk fixed at $0.04$. No outcome missingness was included for simplicity and the simulation design can be found in Web Appendix D. The results given in Figure \ref{fig:simconfound} demonstrate that confounding bias can be significant in rare event settings. Even though the correlation between $A$ and $Y$ is small in such settings, the magnitude of the estimand is also small and thus a small absolute confounding bias can still lead to a large relative bias in the estimates. 

\begin{figure}[ht]
\begin{subfigure}{.5\textwidth}
  \centering
 
  \includegraphics[width=1\linewidth]{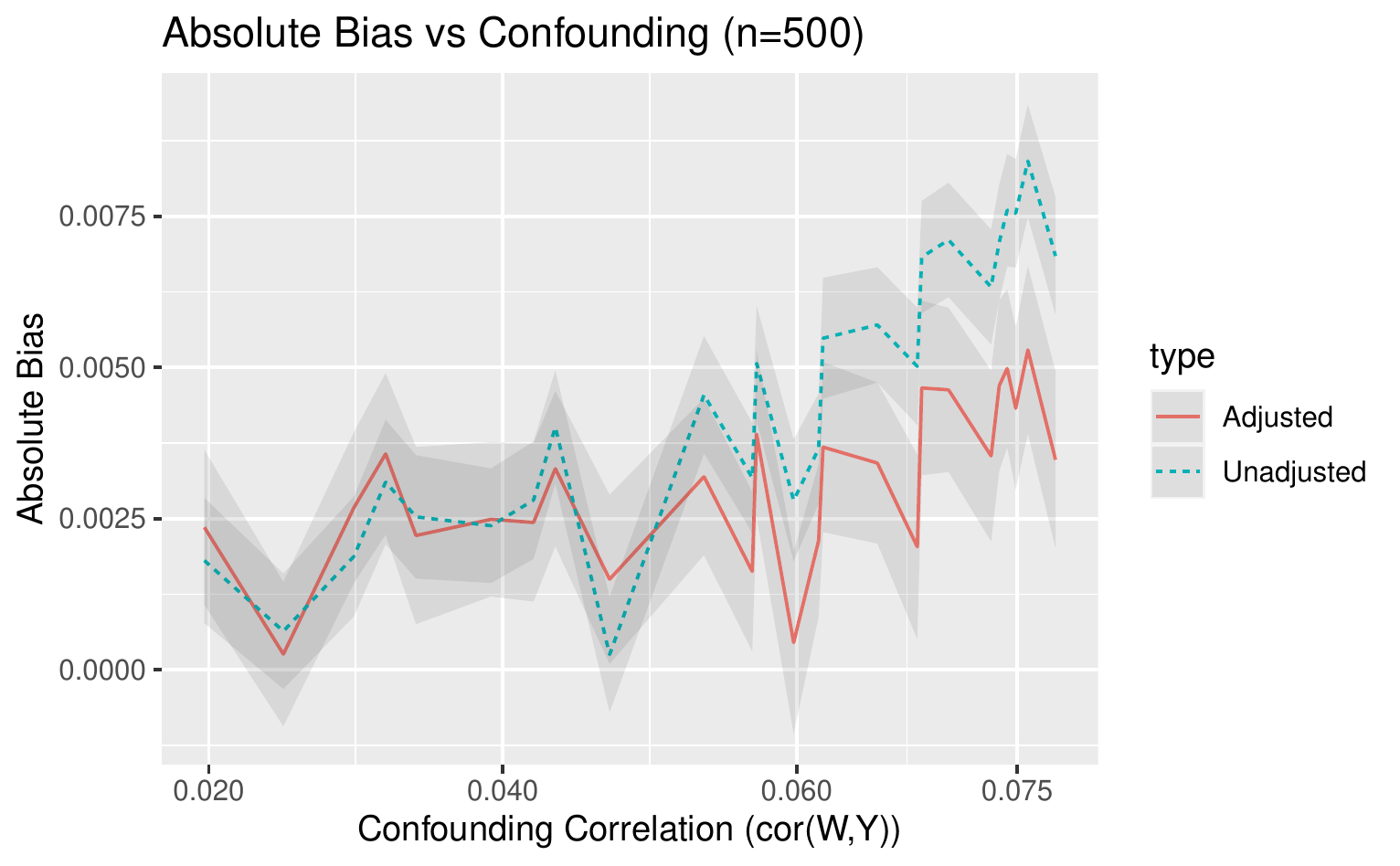}  
 
  \label{fig:sub-first}
\end{subfigure}
\begin{subfigure}{.5\textwidth}
  \centering
 
  \includegraphics[width=1\linewidth]{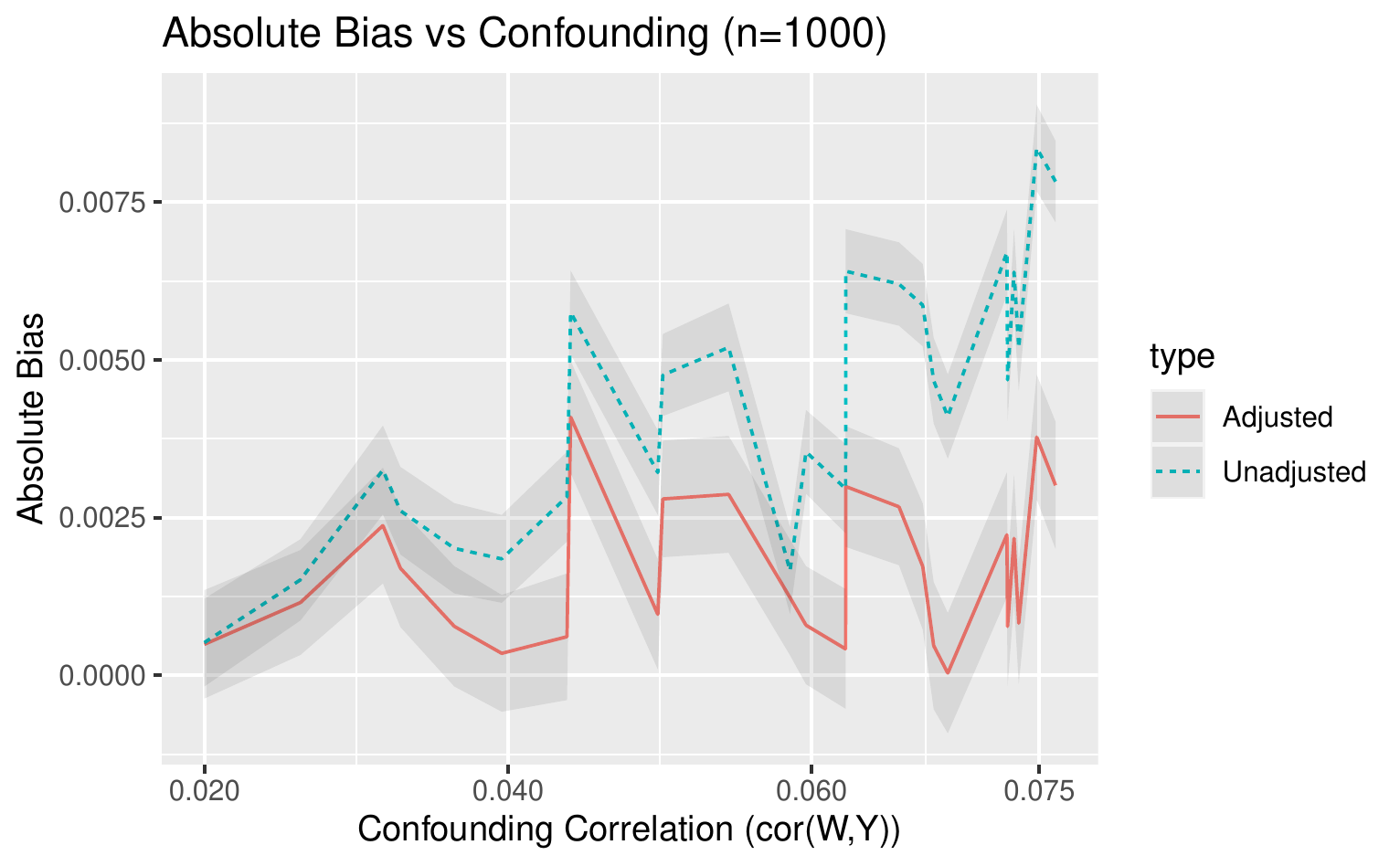}  

  \label{fig:sub-second}
\end{subfigure}
 
\hspace{5cm}\begin{subfigure}{.5\textwidth}
  \centering
  
  \includegraphics[width=1\linewidth]{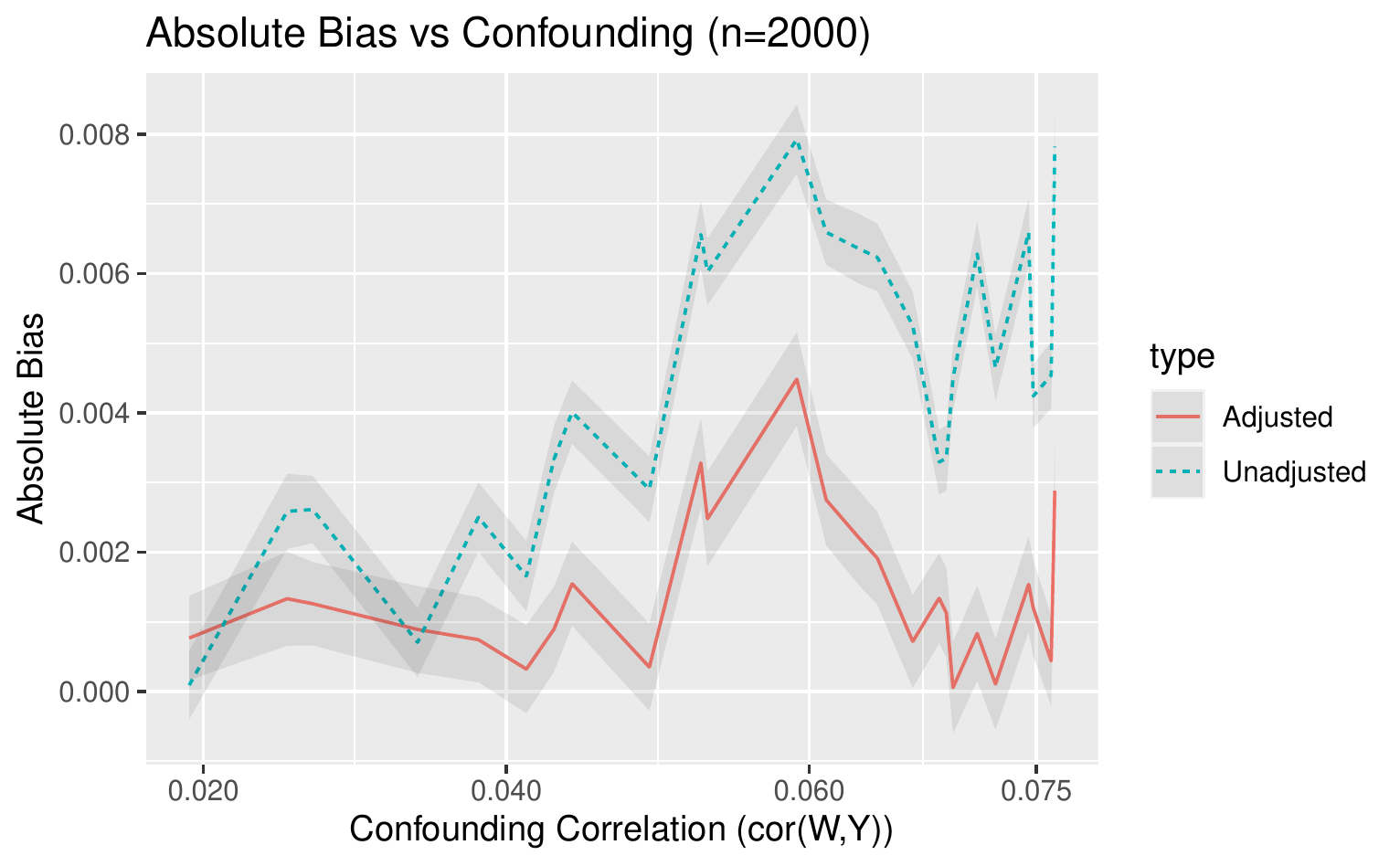}

  \label{fig:sub-second}
\end{subfigure}

\caption{Simulation results depicting relative bias (absolute bias divided by estimand value) for the Donovan estimator and covariate-adjusted srTMLE for various levels of confounding correlation (of the confounder $W$ with $Y$) in rare event settings with average risk $\approx 0.04$. Estimates are based on $500$ monte-carlo estimates at sample sizes $n = 500, 1000, 2000$. Cor(W,A) = 0.5 across all simulations. Uncertainty intervals, defined as $1.95/\sqrt{500}$ times the standard deviation of monte-carlo estimates, are also displayed. This figure appears in color in the electronic version of this article, and any mention of color refers to that version. }
\label{fig:simconfound}
\end{figure}


\section{Application}
We apply the methods developed in this manuscript 
to the same CYD14 and CYD15 dengue vaccine trial data sets analyzed by Donovan, Hudgens, and Gilbert (2019).
CYD14 and CYD15 were Phase 3 placebo-controlled trials that
evaluated the efficacy of the dengue vaccine CYD-TDV in children. CYD14 was conducted in five Asian-Pacific countries with participants between 2 and 14 years old. CYD15 was conducted in five Latin American countries with participants between 9 and 16 years old. The two study designs were harmonized allowing for an analysis of the pooled data, where for pooled analysis we restrict to 9--16-year-olds given the vaccine is approved for this age range. Doses were given at the start of the study (month 0), month 6, and month 12, with follow-up visits at month 13 and month 25. The primary objective assessed vaccine efficacy against the dengue disease primary endpoint occurring between months 13 and 25 in the per-protocol population, where per-protocol was defined as receiving all three immunizations and not experiencing the dengue endpoint between months 0 and 13. Both studies used a case-cohort sampling design where a simple random sample of participants was selected for measurement of anti-dengue neutralizing antibody titers at month 13, augmented with these titer measurements for all per-protocol dengue endpoints (Moodie et al., 2018).\nocite{Moodieetal2018} 
Using the same convention as Donovan et al., the marker of interest, 
month 13 log10 neutralizing antibody titer, is defined
as the average of the four log10 antibody titers to the four serotypes represented inside the vaccine.

We perform the threshold-analysis for both datasets separately and adjust for age, sex, and country. The biomarker variable $A$ is defined as the antibody titer and the outcome $ Y$ is defined as $Y = 1(T \leq t_f)$ where $T$ is the time from the month 13 visit to observed dengue endpoint diagnosis and $t_f$ is a reference time point defined as 336 days after the month 13 visit. Since 99.8\% of participants were evaluable for whether they experienced the dengue endpoint by month 25, we omitted all individuals censored before time $t_f$ from the analysis at a negligible increase in bias. To estimate the adjusted threshold-response function, we apply the srTMLE defined in Section \ref{section::seqregTMLE} with efficient IPW-adjustment to account for the cumulative-case control sampling design (see Web Appendix E for how to adjust the TMLE). The IP-weights and all other nuisance functions are estimated nonparametrically using the Highly Adaptive Lasso (Benkeser, van der Laan, 2016). We also estimate the unadjusted threshold-response function using the IPW-weighted estimator of Donovan et al., and for comparability the same IPW-weights as the srTMLE estimator were used. 
The estimated adjusted (TMLE) and unadjusted (Donovan) threshold-response functions with pointwise 95\% confidence intervals are given in Figures \ref{F1}, \ref{F2} and \ref{F3}. Figure \ref{figure::hists} displays three plots of the reverse-CDF (RCDF) of the immune-response biomarker as a function of the threshold by various covariate strata, and a plot of the estimated expected outcome within levels of the propensity score $P(A \geq v\mid W)$ for the pooled CYD14 + CYD15 analysis.

\begin{figure}
\centering

\hspace{5cm}\begin{subfigure}{.5\textwidth} 
\caption{CYD14 Study }
 \includegraphics[width=1\textwidth]{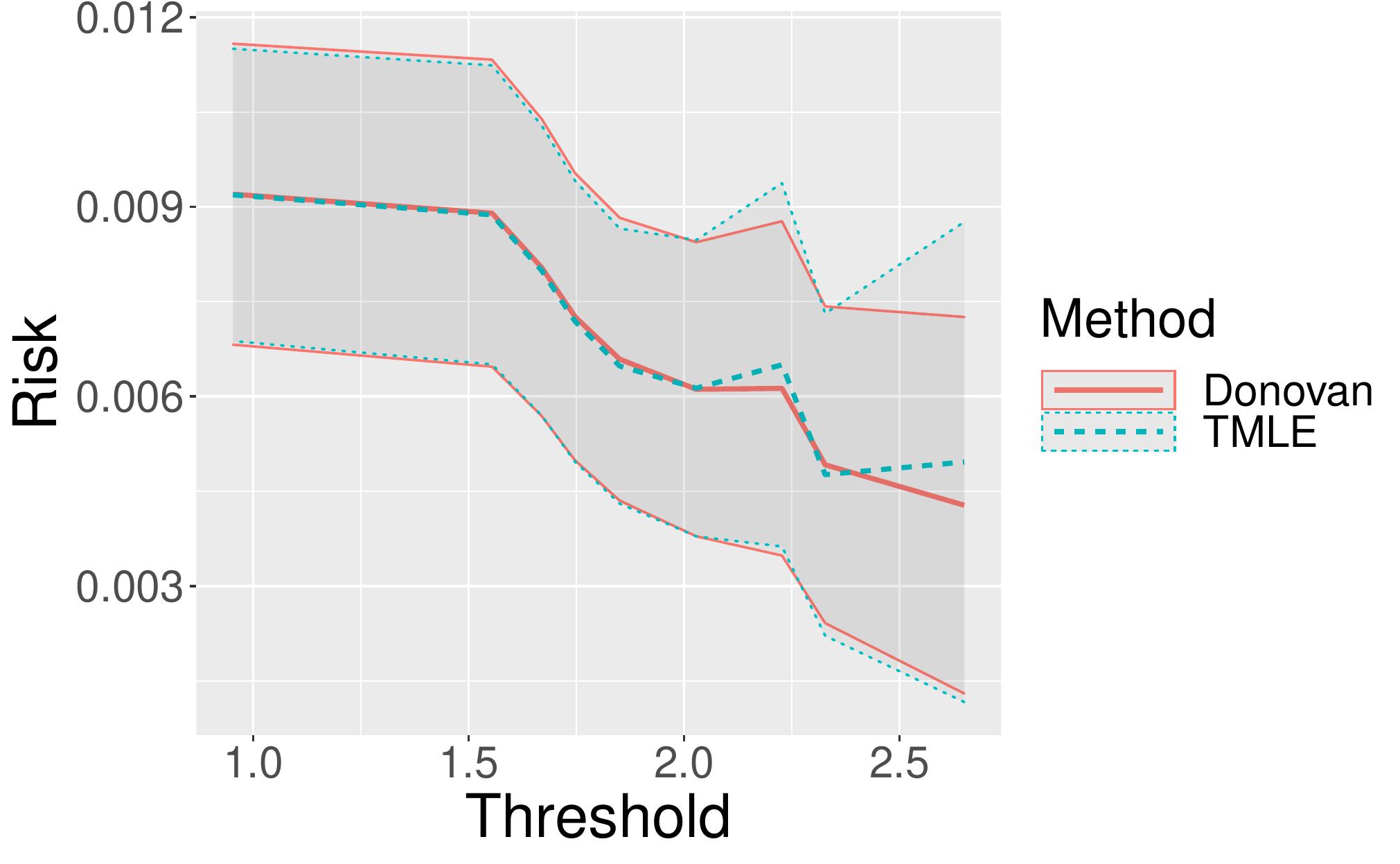} 
\label{F1}
\end{subfigure}\begin{subfigure}{.5\textwidth} 
\caption{CYD15 Study }
 \includegraphics[width=1\textwidth]{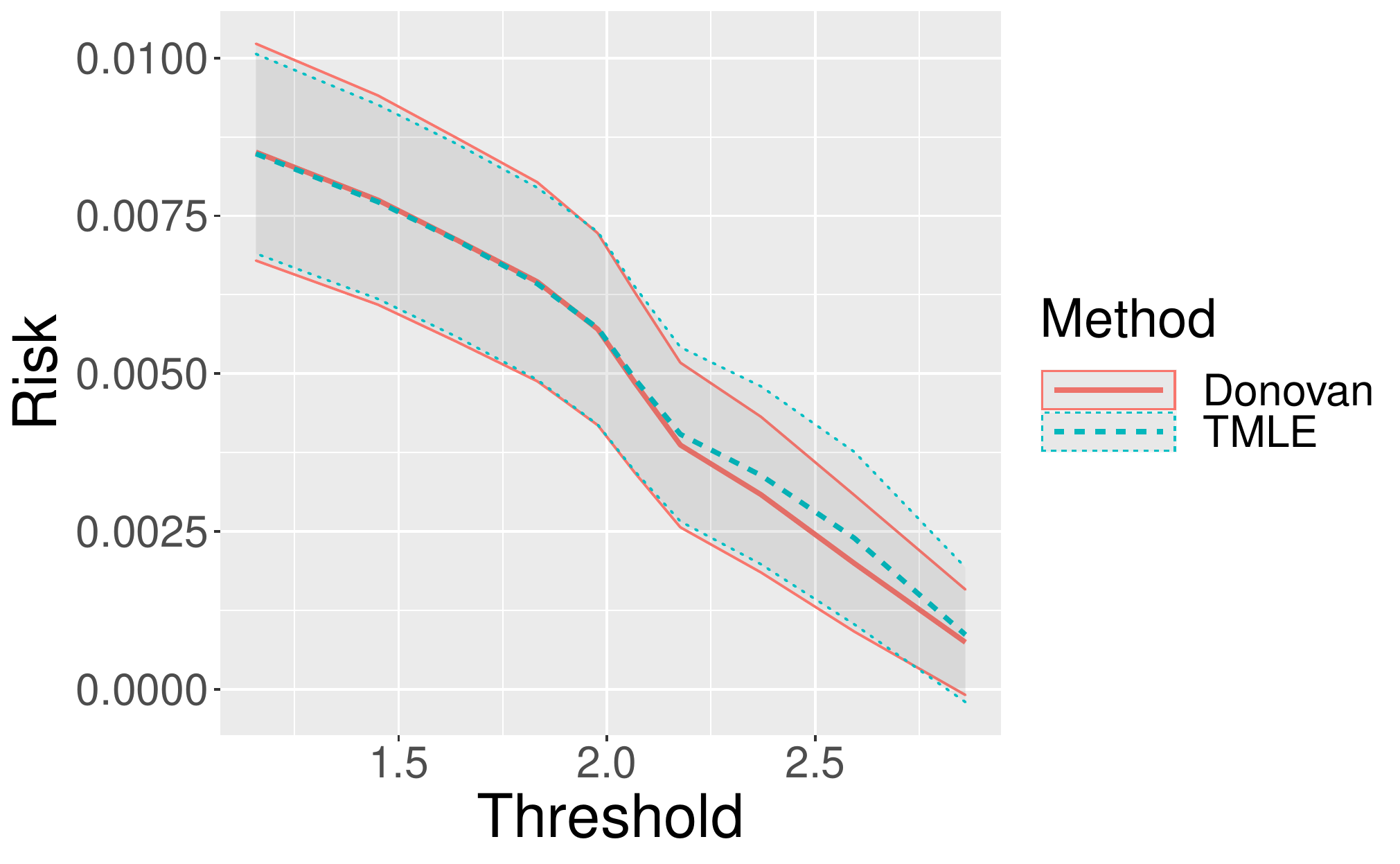}
\label{F2}
\end{subfigure}
\hspace{5cm}\begin{subfigure}{1\textwidth} 
\caption{CYD14 + CYD15 9--16 year-olds Pooled}
\includegraphics[width=0.5\textwidth]{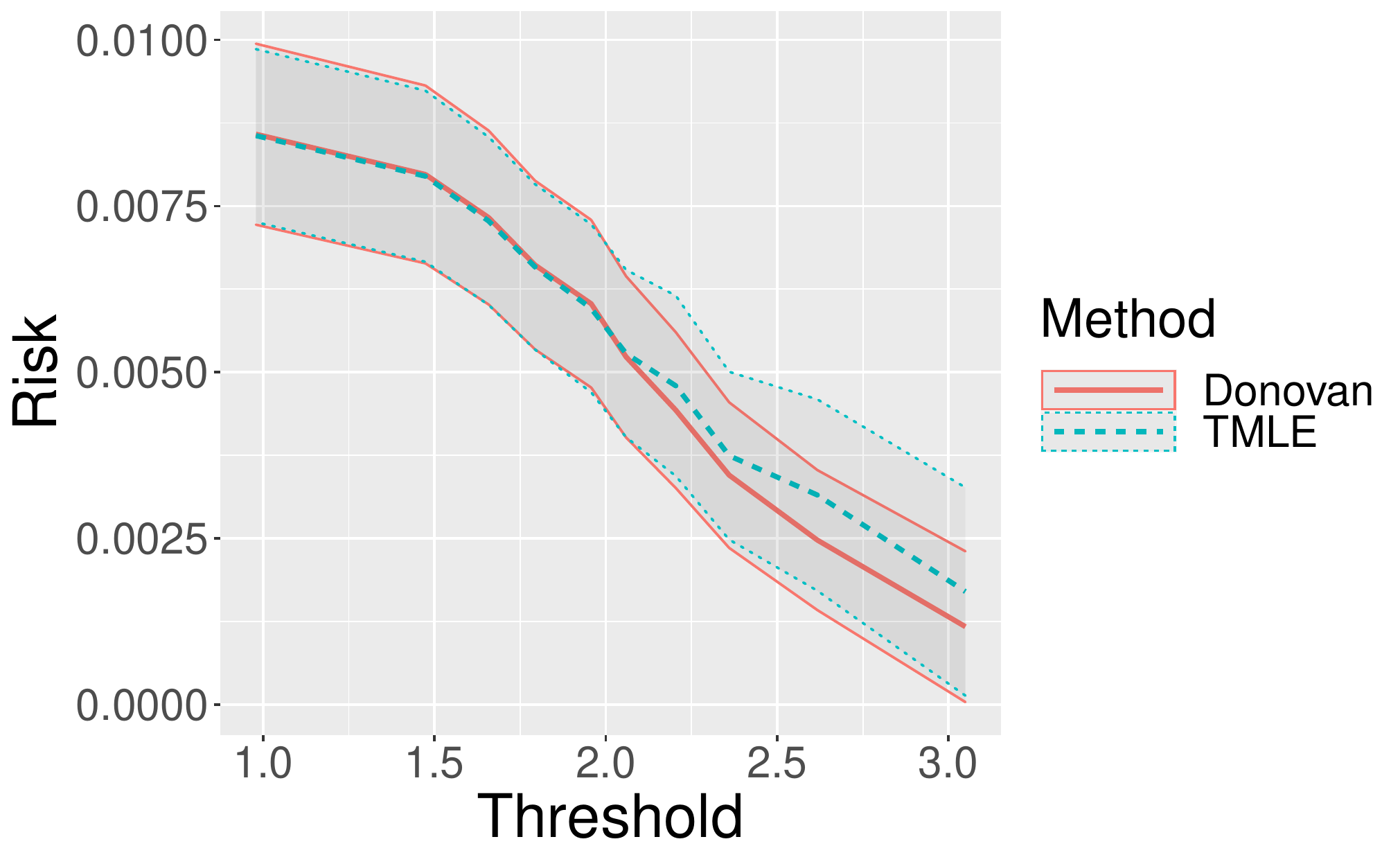}\includegraphics[width=0.5\textwidth]{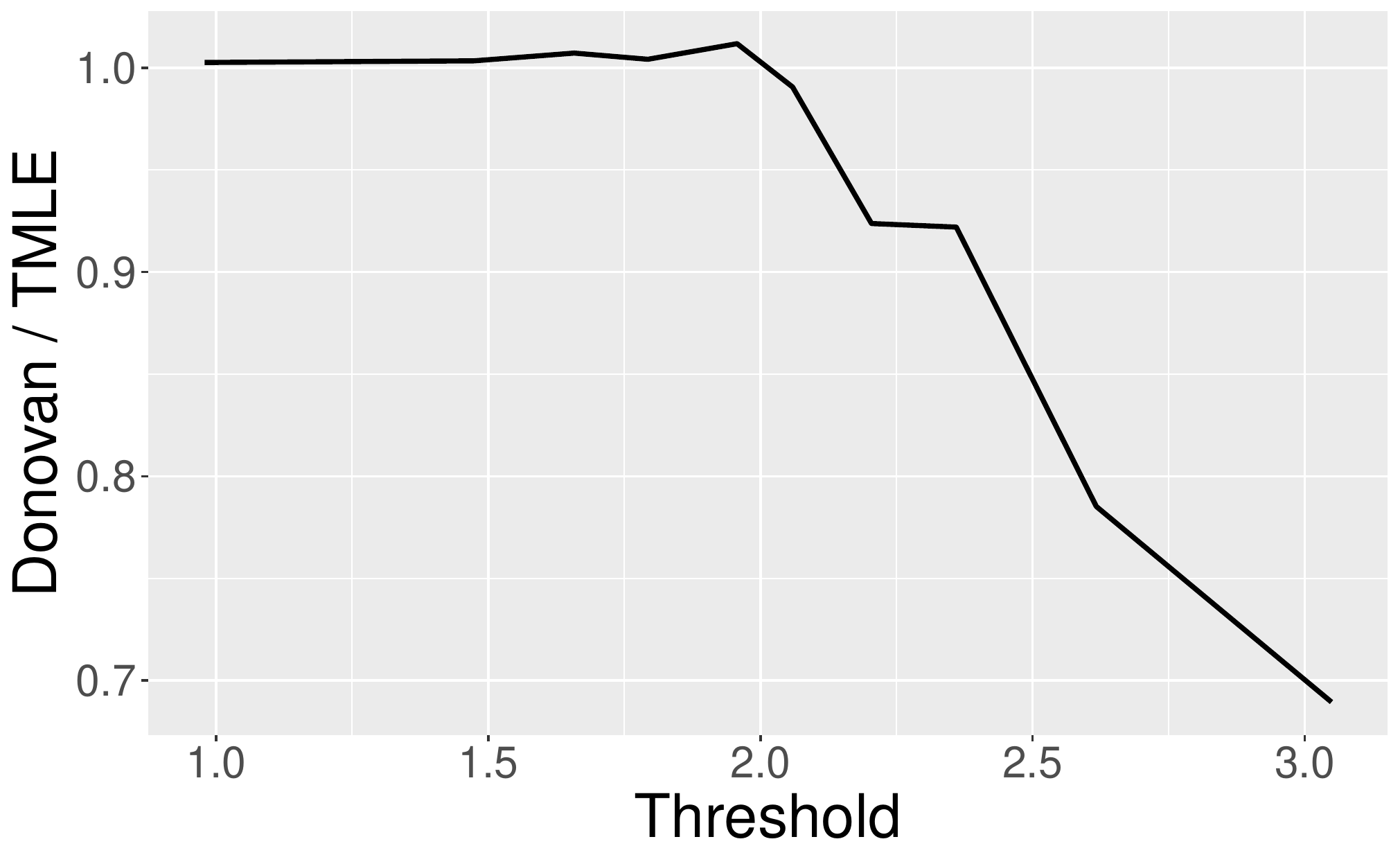}

\label{F3}
\end{subfigure}
\caption*{Adjusted (blue) and unadjusted (red) threshold-response function for a range of thresholds of log10 titer at month 13 for the CYD14 (a), CYD15 (b), and CYD14+15 9--16 years-old (c) pooled studies with pointwise 95\% confidence bands. For (c), the ratio of the Donovan and srTMLE estimates are also plotted. This figure appears in color in the electronic version of this article, and any mention of color refers to that version.
}
\end{figure}

\begin{figure}
\centering

\hspace{5cm}\begin{subfigure}{.5\textwidth} 
 
 \includegraphics[width=1\textwidth]{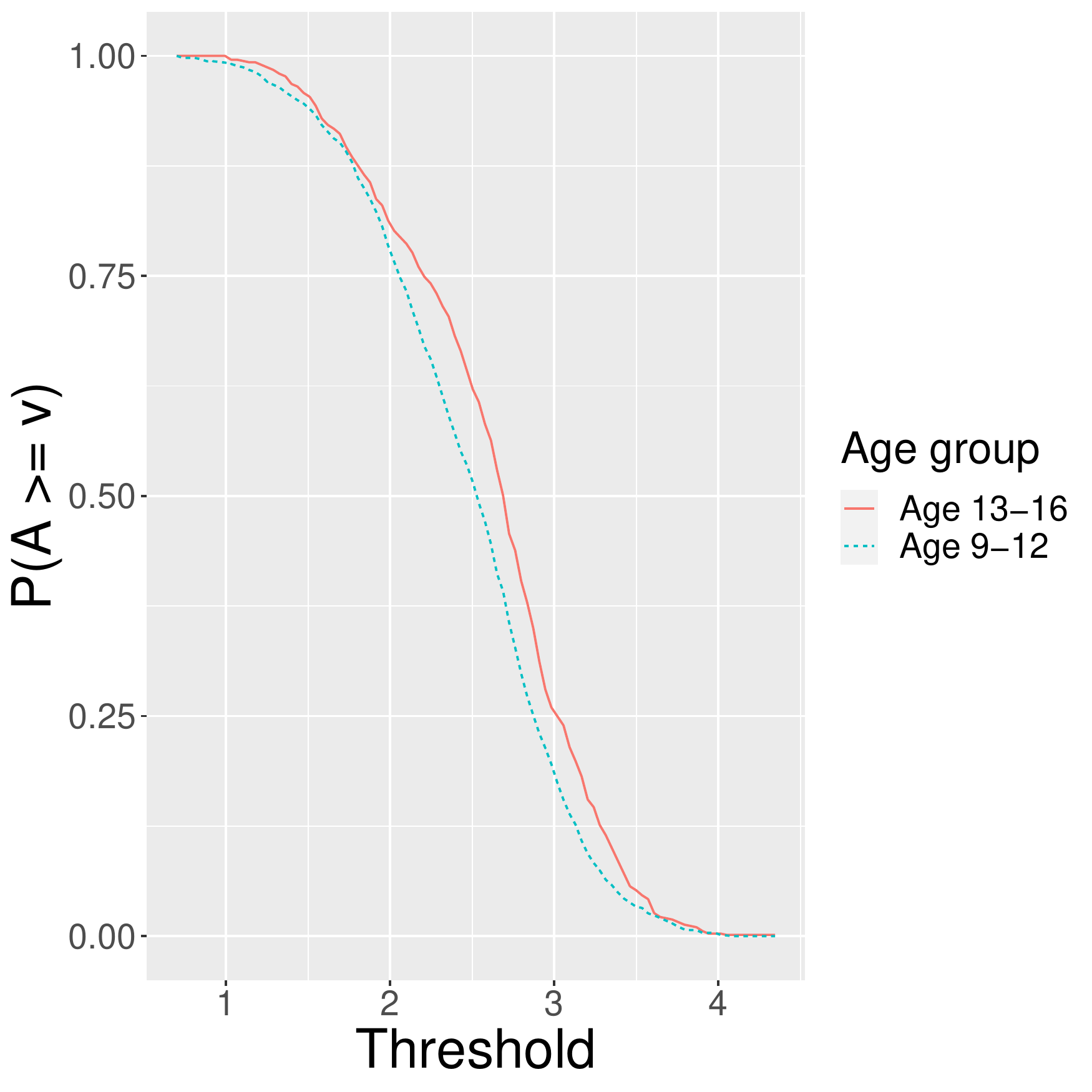}

\end{subfigure}\begin{subfigure}{.5\textwidth} 
 
 \includegraphics[width=1\textwidth]{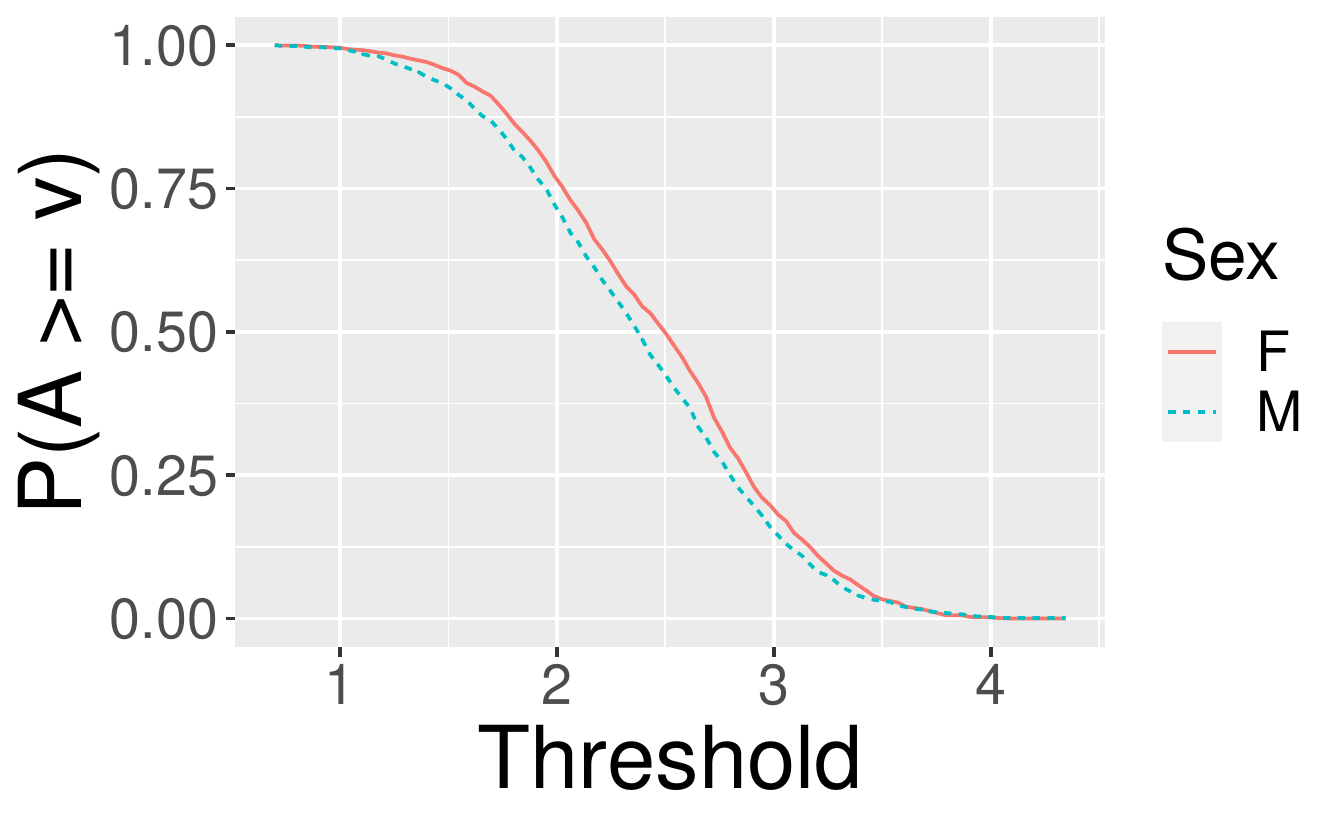} 
 
\end{subfigure}

\hspace{5cm}\begin{subfigure}{.5\textwidth} 
 
 \includegraphics[width=1\textwidth]{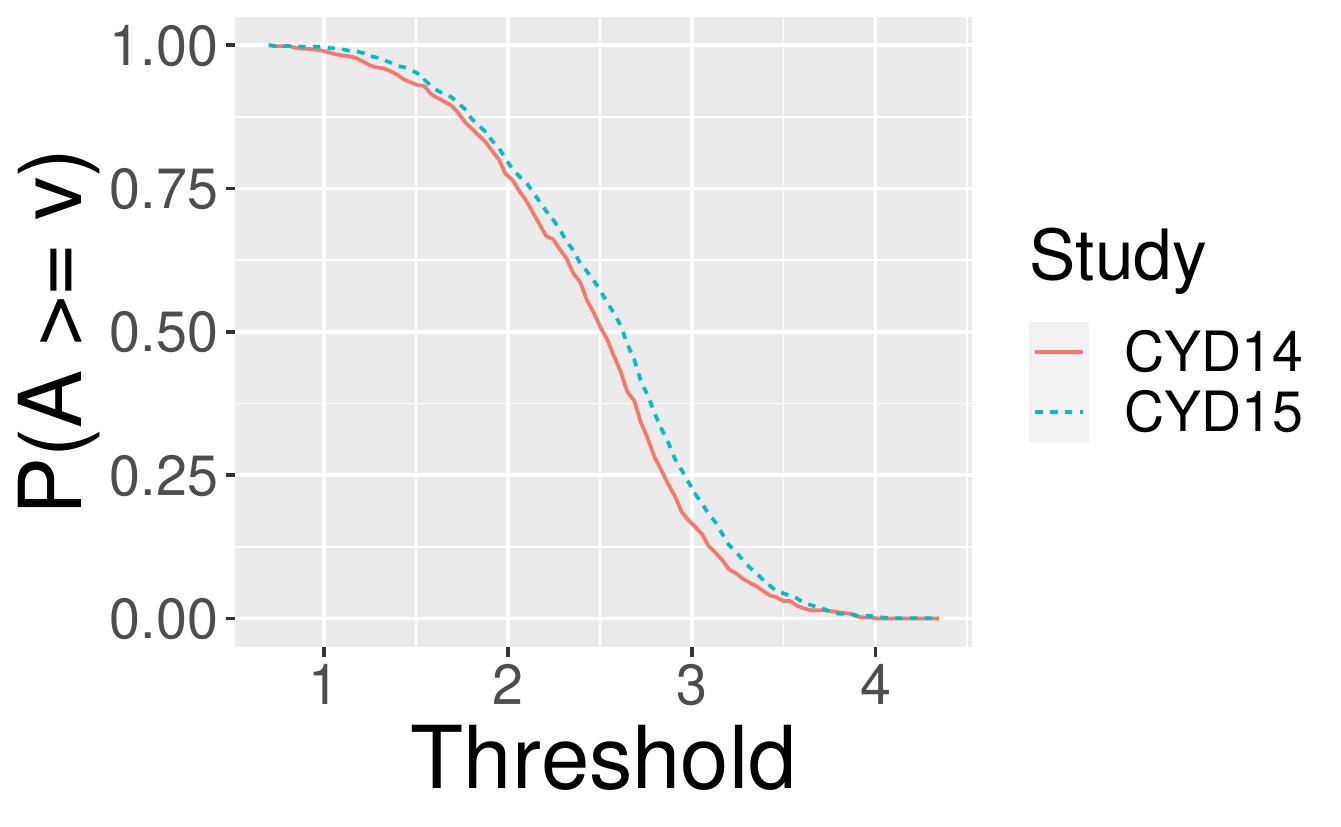}

\end{subfigure}\begin{subfigure}{.5\textwidth} 
 
 \includegraphics[width=1\textwidth]{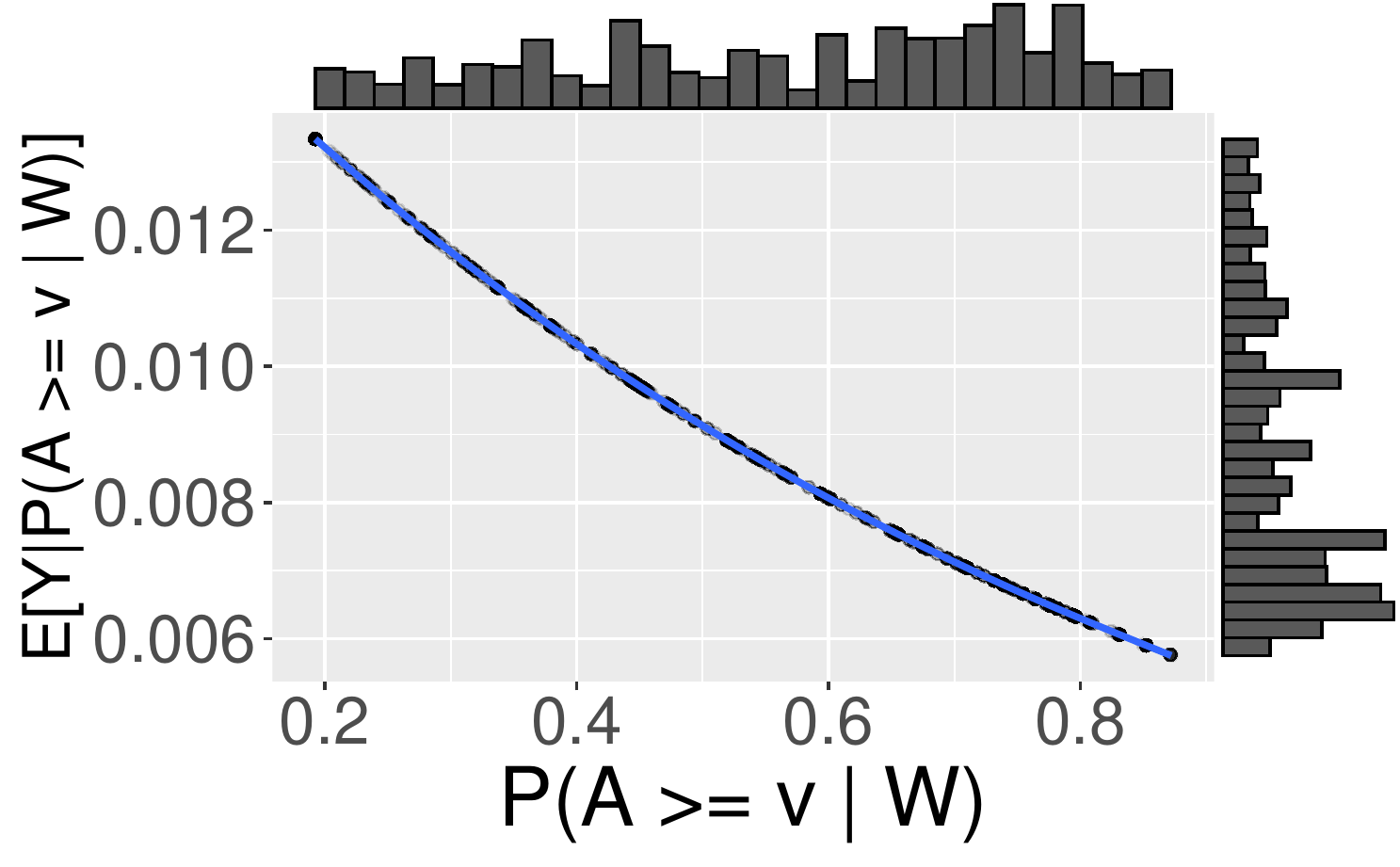} 
 
\end{subfigure}

\caption{(Top + Bottom left) Plots of the reverse-CDF as a function of the biomarker threshold by covariate stratum. (Bottom right) A plot of the estimated expected outcome within levels of the estimated propensity score -- $E_0[Y \mid P(A \geq v\mid W)]$ as a function of the estimated propensity score where $v = \text{Median}(A)$. Marginal histograms of the estimates of $E_0[Y \mid P(A \geq v\mid W)]$ and $P(A \geq v \mid W)$ evaluated at observations pooled over the CYD14 and CYD15 studies are also shown. This figure appears in color in the electronic version of this article, and any mention of color refers to that version.
}
\label{figure::hists}
\end{figure}

We see that month 13 antibody titer is inversely associated with dengue, where subgroups with threshold value above 3.2 (1585 on natural scale) have on average 7-fold lower estimated risk than those with a value above 1.5 (32 on natural scale) (pooled data). There is little-to-no difference between the adjusted and unadjusted threshold-response function estimates for the CYD14 study. On the other hand, for the CYD15 analysis and pooled analysis  there are some noticeable differences between the adjusted and unadjusted threshold-response function estimates for larger thresholds. For the pooled analysis, the unadjusted estimates are around 20-30\% less than the adjusted estimates for the larger thresholds. Figure \ref{figure::hists} (bottom-right) shows that there is reasonable variation of the risk of disease acquisition within levels of the estimated propensity score and that the estimated propensity score values evaluated at observations are close to homogeneously distributed on $[0.2,0.8]$. Figure \ref{figure::hists} (Top and Bottom-left) shows that the distribution of the immune-response biomarker are fairly similar within covariate strata. However, the relative difference in the RCDF becomes more significant as the threshold increases, which suggests the presence of more noticeable confounding bias for larger thresholds. 

     

\section{Conclusions}

In this manuscript, we developed a novel nonparametric efficient and double-robust targeted minimum-loss estimator (srTMLE) for the covariate-adjusted threshold-response function with informative outcome missingness, extending the previous work of Donovan et al. (2019). We also presented an inefficient and possibly biased estimator (binTMLE) for the threshold-response that equals the binary treatment-specific mean TMLE (van der Laan, Rose, 2011) where the continuous biomarker is discretized as above or below a threshold. We showed that the threshold-response at a given threshold can be causally identified with a stochastic intervention whose support is restricted to be above the threshold. The theoretical properties of the novel srTMLE were confirmed in a variety of simulation settings, and it's performance was compared to the binTMLE and the Donovan et al. unadjusted estimator. 
Finally, we applied the method to the CYD14 and CYD15 trial data to estimate the covariate-adjusted threshold-response for associating neutralizing antibody titer in vaccine recipients with dengue disease incorporating IPW to accommodate case-cohort sampled antibody titer data. We compared the results with the unadjusted threshold-response estimator as presented in Donovan et al. (2019). The difference between the unadjusted and adjusted estimates was negligible, although some small differences were observed for the trial-pooled analysis. While in this article we considered estimation of $E_{0,W}E_0[Y \mid  A \geq v, \,W]$, the results can be applied to the parameter $E_{0,W}E_0[Y \mid  A < v,\, W]$ by defining $A := - A$.  Various relative and additive contrasts of threshold-response-type estimands can be estimated using a substitution estimator based on the proposed srTMLE and inference can be obtained from the delta method. A useful extension of our work would be towards estimation of the threshold-specific survival function $E_{0,W}P_0(T > t \mid  A \geq v,\, W)$. While out of the scope of this article, a hazard-based TMLE or sequential-regression-based TMLE for this parameter can be developed using similar arguments as in this manuscript. Also, interesting future work would be the investigation of immune marker surrogate endpoints based on the threshold-response function.

\begin{center}
\textsc{Acknowledgements}
\end{center}
\vspace*{-8pt}
\noindent Research reported in this publication was supported by the National Institute Of Allergy And Infectious
Diseases (NIAID) of the National Institutes of Health (NIH) under Award Number 2R37AI054165. The content is
solely the responsibility of the authors and does not necessarily represent the official views of the NIH.  The authors thank the participants of the CYD14 and CYD15 trials and our SanofiPasteur colleagues who conducted these trials. We would like to thank the reviewers and the editor for their helpful comments which led to numerous improvements to the manuscript.

\vspace*{-14pt}

\begin{center}
\textsc{Data Availability Statement}
\end{center}
\vspace*{-8pt}
\noindent The CYD14 and CYD15 data are available upon request to the sponsor of the studies, Sanofi Pasteur. Data requests can be made through the following url: https://vivli.org.

\vspace*{-14pt}

\bibliographystyle{biom}
\bibliography{ref}
\section*{Supporting Information}
Web Appendices referenced in Sections 2.2, 2.3, 3.1, 3.2, 3.3, 4.1 are available with this paper at the Biometrics website on Wiley Online Library:
(A) Derivation of efficient influence function of target parameter and causal identification.
(B) Proof of efficiency of the TML estimators.
(C) Definition and discussion of Donovan estimator.
(D) Simulation designs.
(E) Adjusting TMLE when biomarker or treatment is missing-at-random.
(F) Some miscellania regarding inference, testing, causal interpretation.
(G) Nuisance parameter estimation and computational considerations.
(H) R code implementing the estimators given in Section 5.
(I) Applying method to bounded continuous outcomes.

\newpage

\section{Web Appendix A}

\subsection{Efficient Influence Function of $\Psi_v^{adj}$ }

The following argument is standard in efficiency theory. We recommend Bickel, Klaassen, Ritov, Wellner (1993) for a comprehensive treatment. We will first derive the efficient influence function for the more general parameter $\Psi_{v}^{adj}$. Afterwards, by setting the $P(\Delta = 1|A, W) = 1$ and $\Delta = 1$, we obtain the EIF of $\Psi_v^{adj}$ in the no outcome-missingness setting.

In order to compute the EIF, we must compute the pathwise derivative of the map $P \mapsto \Psi_v^{adj}(P)$, viewed as a smooth map on the statistical model $\mathcal{M}$ equipped with a non-parametric tangent bundle.

First, we rewrite the parameter as
$$\Psi_{v}^{adj}(P) = E_W \int E[Y| A = s, W, \Delta = 1] P(A=ds |A\geq v, W)$$
$$=  E_W \int E[Y|A =s, W, \Delta = \delta] \frac{1(\delta=1)}{P(\Delta=1|A=s,W)}  \frac{1(s \geq v)  P(A = ds|W)}{P(A\geq v|W)} P(\Delta= d\delta|A=s,W) $$

$$= E_W \int y \frac{1(\delta=1)}{P(\Delta=1|A=s,W)} P(Y=dy| A =s, W) \frac{1(s \geq v) P(A=ds|W)}{P(A\geq v|W)} P(\Delta=d\delta|A=s,W) $$

$$= E_W \int y \frac{1(\delta=1)}{P(\Delta=1|A=s,W)} P(Y=dy| A =s, W) \frac{1(s \geq v) P(A=ds|W)}{\int_v^\infty P(A =da|W) } P(\Delta=d\delta|A=s,W).  $$

Define the tangent spaces,
$$T\mathcal{M}_Y(P) = \left\{(y,a,w, \delta) \mapsto D_Y(y,a,w, \delta): E_P[D_Y|A,W, \Delta] = 0,\, D_Y \in L^2(P) \right\},$$
$$T\mathcal{M}_A(P) = \left\{(a,w) \mapsto D_A(a,w): E_P[D_A|W] = 0,\, D_A \in L^2(P) \right\},$$
$$T\mathcal{M}_W(P) = \left\{(w) \mapsto D_W(w): E_P[D_W] = 0,\, D_W \in L^2(P) \right\}.$$
Let $P \in \mathcal{M}$. Smooth paths $P_{\varepsilon} \subset \mathcal{M}$ with $P_{\varepsilon = 0} = P$  are of the form: 
$$P(Y = dy,\Delta = d\delta, A = da,W=dw) 
$$ $$= P(Y=dy|A = a, W = w, \Delta = \delta)(1+ \varepsilon D_Y(y,a,w, \delta)) P(A=da| W=w)(1+ \varepsilon D_A(a,w))$$
$$\cdot dP(W=dw)(1+ \varepsilon D_W(w)) P(\Delta = d\delta|a, w)$$
for $D_W \in T\mathcal{M}_W(P),\, D_A \in T\mathcal{M}_A(P),\, D_Y \in T\mathcal{M}_Y(P).$ Note, we assume without loss of generality that the missingness mechanism is known, and do not vary the missingness mechanism along these paths. This can be done because $\Psi_{v}^{adj}$ does not depend on the missingness mechanism, and therefore the canonical gradient in the statistical model where the missingness mechanism is known is identical to the canonical gradient in the model where it is not known.
It can be verified that the score of such a path $P_{\varepsilon}$ is exactly $D:= D_Y + D_A + D_W.$
Define $ h(\delta, s, W) = \frac{1(\delta=1)}{P(\Delta=1|s,W)}$ and $G(\delta, s, W) = P(\Delta=\delta|A=s,W)$
Computing the path-wise derivative, we find
$$\frac{d}{d\varepsilon} \Psi_{v , \Delta}^{adj}(P_{\varepsilon}) |_{\varepsilon = 0}=$$
$$
E_W \int y P(Y=dy| s, W, \delta)( D_Y(y, s, W, \delta)) \frac{1(s\geq v) P(A = ds|W) }{\int_v^\infty P(A =da|W)} h(\delta, s, W) G(d\delta, s, W) $$

$$+ E_W \int y P(dy| A =s, W, \delta) \frac{1(s\geq v)  P(A = ds|W) ( D_A(s, W))}{\int_v^\infty P(A =da|W)  }  h(\delta, s, W) G(d\delta, s, W)$$

$$ -E_W \int y  P(Y=dy| A =s, W) \frac{1(s\geq v) P(A = ds|W) }{ \left(\int_v^\infty  P(A =da|W)  \right) ^2} $$
$$\cdot \left( \int_v^\infty  P(A =da|W) ( D_A(a, W))  \right) h(\delta, s, W) G(d\delta, s, W) $$
$$+  E_W \int D_W(w) E[Y|A = s, W, \delta] P(A=ds |A\geq v, W)h(\delta, s, W) G(d\delta, s, W)$$

Define,
$$\Psi_{v,W}^{adj}(W) = \int E[Y| A = s, W, \Delta = 1] P(A=ds |A\geq v, W) .$$
Some simplification of the previous expressions give

$$\frac{d}{d\varepsilon} \Psi_{v}^{adj}(P_\varepsilon) \Big |_{\varepsilon = 0} = 
\int Y  D_Y(Y, A, W, \Delta) \frac{1(A \geq v)}{P(A \geq v |W)}\frac{1(\Delta = 1)}{P(\Delta=1|A,W)} dP$$

$$+ \int \left( E[Y| A, W, \Delta=1] -  \Psi_{v,W}^{adj}(W)\right) \frac{1(A \geq v) D_A(A, W)}{P(A \geq v |W)} dP $$
$$+ \int  \Psi_{v,W}^{adj}(W) D_W(W) dP.$$
Now, define,
$$D_Y^* = \frac{1(A \geq v)}{P(A \geq v |W)}\frac{1(\Delta = 1)}{P(\Delta=1|A,W)}\left[Y - E[Y| A, W, \Delta = 1]\right],$$
$$D_A^* = \left( E[Y| A, W, \Delta = 1] -  \Psi_{v,W}^{adj}\right) \frac{1(A \geq v)}{P(A \geq v |W)}, $$
$$D_W^* =   \Psi_{v,W}^{adj}(W) - \Psi_v^{adj},$$

$$D^*(P) = D_Y^* + D_A^* + D_W^* $$
It is easily verified that $D_Y^* \in T\mathcal{M}_Y(P),\, D_A^* \in T\mathcal{M}_A(P),\, D_W^* \in T\mathcal{M}_W(P), \,D^*(P) \in L^2_0(P).$ In addition, it can be shown that $D_Y^*,  D_A^*$ and $ D_W^* $ are orthogonal in $L^2_0(P).$
It follows that for a path $P_{\varepsilon}$ with score $D(P) \in L^2(P)$, we have
$$\frac{d}{d\varepsilon} \Psi_{v , \Delta}^{adj}(P_\varepsilon) \Big |_{\varepsilon = 0} = \int D(P) \cdot D^*(P) dP.$$
Since, $D^*(P) \in L^2_0(P)$ is a valid score, it follows that $D^*(P)$ is the canonical gradient, and thus the efficient influence function, of $\Psi_{v}^{adj}.$

Now, suppose we have no outcome missingness so that $P(\Delta = 1|A,W) = 1$. Then, $\Psi_{v,W}^{adj}(W) = E[Y|A\geq v,W]$. We then find

$$D^*(P) = \frac{1(A \geq v)}{P(A \geq v |W)}\left[Y - E[Y|A\geq v,W]\right] + E[Y|A\geq v,W] - \Psi_v^{adj}(P).$$

\newpage 
\subsection{Causal Identification of the Parameter}
Define $\widetilde{A}_v = \widetilde{d}_{0,v}(A,W)$. The conditions in the text are equivalent to:
\begin{itemize}
    \item A1. There exists a $\delta > 0$ such that $P(A \geq v |W) > \delta$ a.e. $W$, and $P(\Delta = 1|A,W) > \delta$ a.e. $A, W$ (positivity)
    \item A2. $\widetilde{A}_v \independent  U_Y  |W$
    \item A3. $U_Y \independent  A | W $
    \item A4. $U_Y \independent \Delta| A,W $
    
\end{itemize}
Let $a \geq v$ be in the support of $A$ and let $w$ be in the support of $W$.
$$E \left[Y_{\widetilde{A}_v}| \widetilde{A}_v =a, W = w \right] = E \left[ f_Y( \widetilde{A}_v, W, U_Y) \Big |\widetilde{A}_v = a, W = w\right]$$
$$= E \left[ f_Y(a, w, U_Y) \Big |\widetilde{A}_v = a, W = w\right] = E \left[  f_Y(a, w, U_Y) \Big |W = w\right]$$
$$=  E \left[  f_Y(a, w, U_Y) \Big |A = a, W = w\right] =   E \left[ f_Y(a, w, U_Y) \Big |A = a, W = w, \Delta = 1\right] $$ $$= E \left[   Y \Big |A = a, W = w, \Delta = 1\right] .$$
The first equality  follows by definition of $Y_{\widetilde{A}_v}$ and the second inequality follows from the properties of the conditional expectation. The third inequality follows from independence assumption A2 and the fourth inequality follows from independence assumptions A3. The fifth equality follows from assumption A4 and the conditional expectation with respect to $\Delta$ is well-defined a.e. by the second part of assumption A5.
Therefore,
$$E\left[ Y_{\widetilde{A}_v} \right] = E \left[ f_Y(\widetilde{A}_v, W, U_Y)  \right] = E\left[ E \left[ f_Y(\widetilde{A}_v, W, U_Y) \Big |\widetilde{A}_v, W \right] \right]$$
$$ = \int E \left[ f_Y(\widetilde{A}_v, W, U_Y) \Big |\widetilde{A}_v = a, W = w\right] P_{\widetilde{A}_v}(da|W=w) P_W(dw)$$
$$ = \int  E \left[Y \Big |A = a, W = w, \Delta = 1\right] P_{\widetilde{A}_v}(da|W=w) P_W(dw) = E_W E[E[Y| A, W, \Delta = 1] | A \geq v, W] $$
which concludes the proof of the identification result.
 
 \newpage 

\section{Web Appendix B}

\subsection{Efficiency of the sequential-regression-based TMLE}

\begin{enumerate}[ label=({B}{{\arabic*}})]
    \item  $\delta < P(Y = 1 \mid A \geq v) < 1- \delta$ for some $\delta > 0$.
    \item $g_{0,v}, G_{0} > \delta$ and $g_{n,v}, G_n > \delta$ with probability tending to 1 for some $\delta >0$.
    \item  The set of realizations of $w \mapsto g_{n,v}(w)$, $(a,w) \mapsto G_n(a,w)$, $(a,w) \mapsto Q_n(a,w)$, and $w \mapsto Q_{n,v}(w)$ are $P_0$-Donsker.  
    \item  $\norm{Q_{n} - Q_{0}} = o_P(n^{-1/4})$, $\norm{Q_{n,v} - Q_{0,v}} = o_P(n^{-1/4})$, $\norm{g_{n,v} - g_{0,v}}  = o_P(n^{-1/4})$, $\norm{G_{n} - G_{0}}  = o_P(n^{-1/4})$.
    \item  The union over all $v \in K$ of the set of realizations of the functions given in B3 is $P_0$-Donsker. 
    
\end{enumerate}

 Let $\Psi^*_{n,v} = P_n Q_{n,v}^*$ denote the targeted estimate of $\Psi^{adj}_{v}$.

\begin{theorem}
Suppose conditions B0, B1, B2, B3 and B4 hold. Then, the TML estimator $\Psi^*_{n,v}$ satisfies
$$\sqrt{n} (\Psi^*_{n,v} - \Psi_{v}^{adj}(P_0)) = n^{-1/2} \sum_{i=1}^n D_{P_0,v}(W_i, A_i, \Delta_i, \Delta_i Y_i) + o_p(1).$$
If in addition assumption $C5$ holds, then $\left( \sqrt{n} (\Psi^*_{n,v} - \Psi_{v}^{adj}(P_0)): v \in K\right)$ converges to a tight mean-zero Gaussian process in $l^\infty(K)$ with covariance function $\rho(v_1,v_2) = P_0 D_{P_0,v_1} D_{P_0,v_2}.$

\end{theorem}

\begin{proof}

Define,
$$D_Y^*(g_v,Q, G)(W,A, \Delta, Y) = \frac{1(A \geq v)}{P(A \geq v |W)}\frac{1(\Delta = 1)}{P(\Delta=1|A,W)}\left[Y - E[Y| A, W, \Delta = 1]\right],$$
$$D_A^*(g_v, Q, Q_v)(W,A) = \left( Q(A,W) -  Q_v(W)\right) \frac{1(A \geq v)}{P(A \geq v |W)}, $$
$$D_W^*(E_W, Q_v)(W) =   Q_v(W) - E_W Q_v.$$

The efficient influence function is then given by $D^*(g_v, G, Q, Q_v) = D_Y^*(g_{v},Q, G) + D_A^*(g_{v}, Q, Q_{v}^*) + D_W^*(P_W, Q_{v})$.

The targeting of the sequential regression TMLE ensures that 
$$\frac{1}{n} \sum_{i=1}^n D_Y^*(g_{n,v},Q_{n}^*, G_n)(W_i,A_i,Y_i) = 0,$$
$$\frac{1}{n} \sum_{i=1}^n D_A^*(g_{n,v}, Q_n^*, Q_{n,v}^*)(W_i,A_i,Y_i) = 0.$$
And by estimating $E_W$ with the empirical measure $P_n$, we find
$$\frac{1}{n} \sum_{i=1}^n D_W^*(P_n, Q_{n,v}^*)(W)  = 0.$$

Note that $Q_n^*$, $Q_{n,v}^*$, and $g_{n,v}$ may not be compatible with a single probability distribution $P_{n}^*$, which is why we stress the dependence on each nuisance parameter. Nonetheless since the EIF score equation is solved, we obtain an identical asymptotic expansion as in the proof of Theorem 1. With a slight abuse of notation, we have
$$\Psi(g_{n,v}, G_n, Q_n^*, Q_{n,v}^*) - \Psi(P_0) $$
$$= P_n D^*_v(P_0) + (P_n - P_0) \left[D^*(g_{n,v}, G_n, Q_n^*, Q_{n,v}^*) - D^*_v(P_0) \right] + R_2(g_{n,v}, G_n, Q_n^*, Q_{n,v}^*, P_0),$$
where
$$R_2(g_{n,v}, G_n, Q_n^*, Q_{n,v}^*, P_0) = P_0 \Big[ \frac{1(A \geq v)1(\Delta =1)}{g_{n,v}(W)G_n(A,W)} \left[Q_0(A,W) - Q_n^*(A,W) \right]$$
$$ + \left(Q_{n}^*(A,W) - Q_{n,v}^*(W)  \right) \frac{1(A \geq v)}{g_{n,v}^*(W)} $$
$$+ Q_{n,v}^*(W) - Q_{0,v}(W)  \Big].$$
The remainder can further be written
$$R_2(g_{n,v}, G_n, Q_n^*, Q_{n,v}^*, P_0) =P_0 \Big[ \frac{1(A \geq v)}{g_{n,v}(W)}\left( \frac{G_0(A,W)}{G_n(A,W)}- 1\right) \left[Y - Q_n^*(A,W) \right]$$
$$+ P_0 \Big[ \frac{1(A \geq v)}{g_{n,v}(W)} \left[Y - Q_n^*(A,W) \right]$$
$$ + \left(Q_{n}^*(A,W) - Q_{n,v}^*(W)  \right) \frac{1(A \geq v)}{g_{n,v}(W)} $$
$$+ Q_{n,v}^*(W) - Q_{0,v}(W)  \Big]=$$
$$P_0 \Big[ \frac{1(A \geq v)}{g_{n,v}(W)}\left( \frac{G_0(A,W)}{G_n(A,W)}- 1\right) \left[Q_0(A,W) - Q_{n,v}^*(A,W) \right]$$
$$+ P_0 \Big[\left(1- \frac{g_{0,v}(W)}{g_{n,v}(W)}\right) \left[Q_{0,v}(W) - Q_{n,v}^*(W) \right]. $$

An application of Cauchy-Schwartz and the positivity assumption shows that $$R_2(g_{n,v}, G_n, Q_n, Q_{n,v}, P_0) = O_P \left(\norm{G_n - G_0} \norm{Q_n^* - Q_0}  + \norm{g_{n,v} - g_{0,v}} \norm{Q_{n,v}^* - Q_{0,v}} \right), $$
which is $o_P(n^{-1/2})$ by the consistency assumptions. 

The Donsker assumptions in C2 on all the estimators and the Donsker permanence property (Theorem 2.10.6, van der Vaart, Wellner (1996)\nocite{vanderVaartWellner}) further implies that the empirical process term is $o_P(n^{-1/2}).$
Thus,
$$\sqrt{n} (\Psi^{*}_{n,v} - \Psi_{v}^{adj}(P_0)) = n^{-1/2} \sum_{i=1}^n D_{P_0,v}(W_i, A_i, \Delta_i, \Delta_i Y_i) + o_p(1)$$
as desired.

We now prove, under the regularity assumptions mentioned in Theorem 1, that the TML estimator indexed by the threshold $v \in K$ converges in $l^\infty (K)$ to a tight Gaussian process. Firstly, we claim that the regularity assumptions imply the remainder and empirical process term in the expansion are $o_P(n^{-1/2})$ uniformly in $v$. The uniform convergence of the remainder follows immediately by the assumption of uniform $L^2(P_0)$ convergence over $v \in K$ of $g_{n,v}$ and $Q_{n,v}$ stated in the Theorem. By equi-continuity of the empirical process, the empirical process term claim follows if $\left\{D_{P_n^*, v} - D_{P_0, v}: P_n^*, v \in K \right\}$ is Donsker.  Note that
 both  $\left\{w \mapsto \int_v^\infty g_0(a|w) da: v \in K \right \} $ and \newline $\left\{w \mapsto \int_v^\infty Q_0(a,w) \frac{1}{\int_v^\infty g_0(s|w) ds} g_0(a|w) da: v \in K \right \} $ are Donsker since $v \mapsto \int_v^\infty g_0(a|w)$ and $v \mapsto \int_v^\infty Q_0(a,w) \frac{1}{\int_v^\infty g_0(s|w) ds} g_0(a|w) da$ are Lipschitz, $K$ is bounded and $g_0$ has compact support. \newline Next, it is well known that $\left \{a \mapsto 1(a \geq v): v \in K \right\}$ is Donsker. The Donsker permanence property (Theorem 2.10.6, van der Vaart, Wellner (1996)\nocite{vanderVaartWellner}) implies $\{D_{P_0, v}: v \in K\}$ is Donsker. Now, assumption B5 implies the same for $\{D_{P_n^*, v}: v \in K\}$, and the fact that the difference of two Donsker classes is Donsker proves the claim.

It remains to prove $\sqrt{n} P_n D_{P_0, v}$ converges to a tight mean-zero Gaussian process with the desired covariance function. By standard CLT, we have that $\sqrt{n} \left(\Psi_v^{adj}(P_n^*) - \Psi_v^{adj}(P_0) \right)$ converges to a multivariate normal for all finite subsets $V \subset K$ with the desired covariance matrix. If $K$ is a finite set, then we are done. Otherwise, the functional convergence as a stochastic process in $l^\infty(K)$ follows if the set of functions $\left \{D_{P_0, v}: v \in K \right\}$ is Donsker (van der Vaart, Wellner, 1996), which we have already shown.

\end{proof}

\newpage

\section{Web Appendix C}

\subsection{Unadjusted threshold-response and Donovan's estimator NPMLE}

For the unadjusted parameter $\Psi_v^{unadj}$, the efficient influence function is
$$\tilde D_{P,v}(A, Y) = \frac{1(A\geq v)}{P(A \geq v)} (Y - E_P[Y | A \geq v]),$$
and the efficient influence function of the parameter $\Psi_{v}^{unadj}$ is given by
$$ D_{P,v}(A, \Delta, \Delta Y) =  \frac{1(A \geq v)}{P(A \geq v )}\frac{1(\Delta = 1)}{P(\Delta=1|A)}\left(Y -E_P[Y| A, \Delta = 1]\right) $$
$$+\left( E_P[Y| A, \Delta = 1] -  \Psi_{v}^{unadj}(P)\right) \frac{1(A \geq v)}{P(A \geq v)} .$$

In this section, we will present the nonparametric minimum loss (NPML) estimator for the threshold-response function as described in Donovan et al (2019). Since the estimator is an NPMLE, it solves a wide range of score equations. In fact, the NPMLE may be viewed as a TMLE. In this case, the parameter of interest is the unadjusted threshold-response function, 
$\Psi_v^{unadj}(P) = E[Y | A \geq v].$
The Donovan estimator is given by the NPMLE
$$\Psi_v^{unadj} = Q_{n,v}^{unadj}= \frac{\sum_{i=1}^n Y_i 1(A_i \geq v)}{\sum_{i=1}^n 1(A_i \geq v)}.$$
We claim that the Donovan estimator solves the efficient score equation:
$$\frac{1}{n} \sum_{i=1}^n \frac{ 1(A_i \geq v)}{g_{n,v}^{unadj}} \left(Y_i - Q_{n,v}^{unadj} \right)=0,$$
where $g_{n,v}^{unadj} = \frac{1}{n} \sum_{i=1}^n 1(A_i \geq v).$
To see this, note that $g_{n,v}$ and $Q_{n,v}$ are constant, since they are not functions of the baseline covariates $W_i$. Thus,
$$\frac{1}{n} \sum_{i=1}^n \frac{ 1(A_i \geq v)}{g_{v,n}^{unadj}} \left(Y_i - Q_{v,n}^{unadj} \right) =  \left[ \frac{1}{n} \sum_{i=1}^n  \frac{1(A_i \geq v)Y_i}{g_{v,n}^{unadj}} \right] -  Q_{v,n}^{unadj}  = 0.$$
As a result, one may view the Donovan estimator as a TMLE for the covariate-adjusted threshold-response function where the nuisance estimates are $g_{n,v}^{unadj}$ and $Q_{v,n}^{unadj}$. In this case, the parametric update performed in the targeting step does not change the initial estimator.

Since we may view the estimator as a special case of a TMLE, we find that in the nonparametric statistical model where no baseline covariates are utilized, the Donovan estimator is an efficient estimator for the unadjusted threshold response function $E[Y| A \geq v]$. In the more general case of no confounding baseline covariates, one can show that the Donovan estimator is consistent, but not efficient, for the covariate-adjusted threshold response function $E_W E[Y|A\geq v, W]$. The Donovan estimator is always consistent for $E[Y| A \geq v]$.

In the case of outcome missingness, the Donovan estimator can be extended as
$$\hat \Psi_{v,\Delta,n}^{unadj}= \frac{\sum_{i=1}^n Y_i \Delta_i 1(A_i \geq v)}{\sum_{i=1}^n \Delta_i 1(A_i \geq v)}.$$
However, this estimator is generally not efficient and is only consistent when the missingness mechanism does not depend on $A$. More generally, this is an estimator for $E[Y|A \geq v, \Delta = 1].$

We note that the Donovan estimator can be extended to right-censored survival data using the Kaplan Meier estimator (Donovan et al (2019)).

The following corollary follows from an application of Theorem 1.
\begin{corollary}
Suppose conditions B0 and B1 hold. Then, the Donovan NPML estimator satisfies
$$\sqrt{n} (\hat \Psi_v^{unadj} - \Psi_v^{unadj}(P_0)) = n^{-1/2} \sum_{i=1}^n \tilde D_{P_0,v}(A_i, Y_i) + o_p(1),$$
and converges to a tight mean-zero Gaussian process in $l^\infty(K)$ with covariance function $\rho(v_1,v_2) = P_0 \tilde D_{P_0,v_1} \tilde D_{P_0,v_2}.$

\end{corollary}
It follows that the Donovan NPML estimator is an efficient estimator for $ \Psi_v^{unadj}(P_0)$ when there are no baseline covariates.

\newpage

 \section{Web Appendix D}
 
 \subsubsection{Simulation design 1}
 
The baseline variables were generated as:
 $$W_1 \sim truncnorm(a = -0.75, b = 1.5, mean = 0.5, sd = 0.75$$
 $$W_2 \sim bernoulli(p = 0.6),$$
  $$W_3 \sim bernoulli(p = 0.3).$$
  \subsubsection{Simulation design 2}
  The baseline variables were generated as:
   $$W_1 \sim truncnorm(a = -0.75, b = 1.5, mean = 0.5, sd = 0.75$$
 $$W_2 \sim bernoulli(p = 0.6),$$
  $$W_3 \sim bernoulli(p = 0.3).$$
  All nuisance functions were estimated using generalized additive models (additive splines). 
\subsubsection{Simulation on effect of covariate adjustment in reducing confounding bias}
For parameters $c \in [0,2.5]$ and normalization parameter $K_c > 0$, we simulate the data structure $(W_1, W_2, A, Y)$ as follows.
$$W_1 \sim \text{Unif}(0,1)$$
$$W_2 \sim \text{Unif}(0,1)$$
$$A \sim \text{Normal}(\mu = -0.6W_2, \sigma = 0.3)$$
$$Y \sim \text{Bernoulli}(p =  \textbf{$K_c$}*0.1 *\text{expit}(-1 - 1.3*A - \exp( A) - 2*A^2 + W_1 - 0.25*A*W_1 + (1)\textbf{c}*W_2)).$$
 $W_2$ plays the role of the confounder and satisfies $Cor(W_2,A) \approx 0.5$. For each value of $c$, $K_c$ is chosen so that $P(Y=1) \approx 0.04$, so that the level of outcome rareness stays constant while confounding varies. The coefficient $c$ plays the role of confounding level. As $c$ increases, the confounding correlation $Cor(W_2,Y)$ increases accordingly. For each level of $c$, we estimate $\Psi_v^{adj}(P_0)$ with $v := \text{median}(A)$ using the unadjusted TMLE (Donovan estimator) and the adjusted TMLE. For the adjusted TMLE, we utilized gradient boosting with maximum depth $5$ to estimate the nuisance parameters. There was no missingness or biased sampling in $A$, so no IPW-based adjustment was performed. For each simulation setting, we performed $500$ monte-carlo simulations.

\subsubsection{Simultaneous confidence interval coverage}
The data structure $(W_1,W_2,A,Y)$ was generated as follows.
$W_1 \sim \text{Unif}(0,1)$, $W_2 \sim \text{Unif}(0,1)$, 
$A \sim \text{Normal}(\mu = -0.6W_2, \sigma = 0.3)$ and $Y \sim \text{Bernoulli}(p =  0.032764*0.1 *\text{expit}(-1 - 1.3*A - \exp( A) - 2*A^2 + W_1 - 0.25*A*W_1 + (1)\textbf{0.9375}*W_2))$ with
$$p=  0.032764*0.1 *\text{expit}(-1 - 1.3*A - \exp( A) - 2*A^2 + W_1 - 0.25*A*W_1 + (1)\textbf{0.9375}*W_2).$$

We estimated the nuisance parameters using gradient-boosting with the maximum tree depth selected by cross-validation. We performed 500 monte-carlo simulations with $n = 2000$ where we computed the TMLE for the threshold-response parameter for the $0,0.1,0.2,0.3,0.4$ and $0.5$ quantiles of $A$. To prevent positivity violations, which would hurt coverage, we did not estimate thresholds corresponding with quantiles above $0.5$. $95\%$ pointwise and simultaneous confidence interval coverage probabilities are displayed below.

\begin{table}[]
    \centering
    \begin{tabular}{c||c}
        Threshold Quantile & Pointwise coverage  \\
        \hline   \hline  \\
         0  & 0.954 \\  \hline  
         0.1 & 0.96\\  \hline  
        0.2  & 0.96\\  \hline  
        0.3  & 0.95  \\  \hline  
        0.4  & 0.97 \\ \hline  
        0.5 & 0.95  \\ \hline  
    \end{tabular}
    \vspace{0.5cm} 
    \caption{The table displays pointwise confidence interval coverage computed from $500$ monte-carlo estimates. Simultaneous confidence interval coverage was found to be 0.948.}
    \label{tab:my_label}
\end{table}

\section{Web Appendix E}

\subsection{Adjusting the TMLE for biased sampling designs}

In this section, we consider the setting where the treatment or marker $A$ is missing/unmeasured for a subset of the observations. In clinical trials, biomarker variables are often only measured in a subset of the study participants. In particular, clinical trials often employ biased sampling designs such as cumulative-case-control sampling or stratified two-stage sampling. The proposed TML estimator and inference method for the threshold-response function can be adjusted using inverse-probability weighting (IPW) as described in Rose and van der Laan (2011) \nocite{vanderLaanRose2011} to account for the biomarker missingness. 
To this end, we consider the more general data structure $O = (W, R, RA, \Delta, \Delta Y) \sim P_0$. $R$ is a binary random variable that takes the value $1$ if the random variable $A$ is observed/measured and $0$ otherwise. For cumulative case-control sampling, we can take $R$ to be a Bernoulli random variable that takes the value $1$ with probability $P_0(R=1|\Delta Y)$. Under the assumption that $A$ is missing-at-random (MAR) given the fully observed data, one can show that $\Psi_{v}^{adj}(P_{0})$ is identified by the data generating distribution of the observed data. Specifically, we assume $A$ is independent of $R$ conditional on $( W, \Delta, \Delta Y)$. This assumption is necessarily satisfied for the cumulative case-control and for two-stage stratified sampling designs that sample based on $(W, \Delta, \Delta Y)$. We define the inverse-probability weights (IPW), $w_0(O) = \frac{1}{P_0(R=1|O)}$. These weights can be estimated from the data by estimating the conditional probability $P_0(R=1|O)$, possibly using nonparametric minimum-loss estimation. 

When there is biomarker missingness, we can apply an inverse-probability weighted version of the srTMLE of the main text. Specifically, for the estimation of the nuisance parameters, we perform the regressions in the subset of the data where $R=1$ and adjust for the treatment missingness by incorporating the weights $w_0(O_i): i=1,\dots,n$ in the regressions. Similarly, in the targeting step where we perform minimum-loss estimation over a parametric fluctuation submodel, we perform the loss minimization in the subset of the data with $R=1$ and incorporate the weights $w_0(O_i): i=1,\dots,n$. The final estimate of $\Psi_v^{adj}(P_0)$ is then given by the substitution estimator where averaging over $W$ is performed using the $w_0$-weighted empirical distribution among observations with $R=1$. As shown in  Rose and M. van der Laan (2011), for efficiency, the initial estimates of the weight function $w_0$ should be targeted. We refer to Rose and M. van der Laan (2011) for an in-depth treatment of this adjustment in the context of TMLE.  

We will now give a fully-efficient biomarker-missing version of the srTMLE given in the main text. Let $w_n$ be an initial estimator of $w_0$ and $D_{n,v}$ an initial estimator of the efficient influence function $D_{P_0, v}$ for the non-biomarker-missingness case. Let $H_n(W, \Delta, \Delta Y)$ be an initial estimator of $E[D_{P_0, v}(O) \mid R =1, W, \Delta, \Delta Y]$, which can be obtained by performing the pseudo-outcome regression of $D_{n,v}$ on $( W, \Delta, \Delta Y)$ using only the observations with $R=1$.

\textbf{Step 1. Target IP-weights }
\begin{enumerate}
    \item Define the logistic submodel $w_{n, \varepsilon}(O) = \text{expit} \left\{\text{logit}(w_n)(O) + \varepsilon \frac{H_n(W, \Delta, \Delta Y)}{w_n(O)} \right\}$
    \item Define the targeted weight function $w_n^* := w_{n , \hat \varepsilon_n}$ where
    $$\hat \varepsilon_n = \argmax_{\varepsilon \in \mathbb{R}} \frac{1}{n} \sum_{i=1}^n R_i \cdot \log w_{n, \varepsilon}(O_i) + (1-R_i) \cdot \log (1-w_{n, \varepsilon}(O_i)) $$
    is the MLE along the submodel.
    \item For an event $\mathcal{A}$, define
    $$P_{W,n}^*(\mathcal{A}) = \frac{1}{n}\sum_{i=1}^n R_iw_n^*(O_i) 1(W_i \in \mathcal{A})$$
    to be the targeted IPW estimate of $P_{W,0}$.
\end{enumerate}

\textbf{Step 2. IPW-TMLE}
\begin{enumerate}
    \item Define the indicator fluctuation submodel
$Q_{n, \varepsilon} (A,W)= \text{expit} \left \{\text{logit}(Q_{n})(A,W) + \varepsilon 1(A \geq v)  \right\}. $ 
    \item The MLE along this submodel is given by
    $\hat \varepsilon_n = $
    $$\argmax_{\varepsilon \in \mathbb{R}}  \sum_{i=1}^n  \frac{R_iw_n^*(O_i) \Delta_i }{ g_{n,v}(W_i) G_n(A_i,W_i)} \left \{ Y_i \cdot \log Q_{n, \varepsilon}(A_i, W_i) + (1- Y_i) \cdot \log (1 - Q_{n,  \varepsilon}(A_i, W_i)) \right\}.$$
    \item Define the updated estimate of $Q_0$ as $Q_{n}^* = Q_{n, \hat \varepsilon_n}$.
    \item Obtain an initial estimator $Q_{n,v}$ of $Q_{0,v}(W) = E_{P_0}[E_{P_0}[Y \mid A, W, \Delta = 1] \mid  A \geq v, W] $ using sequential regression (e.g. estimate $E[Q_n^*(A,W) \mid  A \geq v, W]$).
    \item Define the intercept fluctuation submodel,
$Q_{n, v, \varepsilon}= \text{expit} \left \{\text{logit}(Q_{n,v}) + \varepsilon  \right\}. $
    \item The MLE along this submodel is given by $Q_{n,v}^* = Q_{n,v, \hat \varepsilon_n}$ where
    $\hat \varepsilon_n =$
$$\argmax_{\varepsilon \in \mathbb{R}}   \sum_{i=1}^n  \frac{R_iw_n^*(O_i)1(A_i \geq v)}{ g_{n,v}(W_i)} \left \{ Q_{n}^*(A_i,W_i) \log Q_{n, v, \varepsilon}(W_i) + (1- Q_n^*(A_i,W_i)) \log (1 - Q_{n,  v, \varepsilon}(W_i)) \right\}.$$
    \item Let $P_{n,v}^* := (P_{W,n}^*, g_{n,v}, G_n, Q_{n,v}^*, Q_n^*)$. The TMLE of $\Psi_{v}^{adj}(P_0)$ is then given by the substitution estimator
$\Psi_{v}^{adj}(P_{n,v}^*) = E_{P_{W,n}^*}  Q_{n,v}^*(W)= \frac{1}{n} \sum_{i=1}^n  R_iw_n^*(O_i) Q_{n,v}^*(W_i).$

\end{enumerate}

Inference for TMLE can be obtained as described in the main text except $D_{P_0,v}$ needs to be replaced with the new efficient influence function
$$\widetilde{D}_{P_0, v}(O) := R w_0(O) D_{P_0,v}(O) - \left( R w_0(O) - 1 \right) E\left[D_{P_0,v}(O) \mid R=1, W, \Delta, \Delta Y \right].$$

\subsection{Simulation comparing fully efficient IPW TMLE vs inefficient IPW TMLE}

 As discussed in Rose and M. van der Laan (2011), Step 1 above can be omitted if $w_0$ is estimated with a stratified empirical mean estimator or a correctly-specified parametric model at a possibly significant cost in statistical efficiency. For instance, if the biased sampling is due to cumulative case-control sampling then we have $w_0(O) \equiv w_0(\Delta Y) = P(R=1 \mid \Delta Y)$ is only a function of $\Delta Y$. It would then be straightforward to use the estimator $w_n(\delta y) = \frac{1}{\sum_{i=1}^n 1(\Delta_i Y_i = \delta y)} \sum_{i=1}^n R_i 1(\Delta_i Y_i = \delta y)$.  In this case, inference can be obtained based on the IPW-influence function $\frac{R}{w_0(O)}D_{P_0,v}(O)$. If Step 1 is omitted and Step 2 is performed with the stratified empirical mean weight estimator then the resulting IPW-srTMLE will be statistically inefficient. In the following simulation results, we show that the loss in efficiency can be striking by omitting step 1. and using the non-targeted weights. 
 
 For the simulation, the data-structure is $(W_1, W_2 ,W_3, R, RA, Y)$ and is generated as follows. $W_1, W_2 \sim \text{Uniform}(-1,1)$, $W_3 \sim \text{Exp}(1)$, $A \sim \text{Gamma}(shape = 3,\, rate = 13)$, and the outcome variable $Y$ is Bernoulli-distributed and takes the value 1 with probability,
$$P(Y = 1|A, W) =  \text{expit}(-4.7+
    0.7 \cdot (
      0.7\cdot(0.5-1.5\cdot A + 0.25\cdot(W_1 + W_2 + W_3)$$
      $$+ \sin(3\cdot W_1) + \sin(3\cdot W_2)+ \log(1+W_3)$$
      $$+ 2\cdot W_1\cdot \sin(3\cdot W_2) + 2\cdot W_2\cdot \sin(3\cdot W_1) + W_3\cdot \sin(3\cdot W_1) + W_3 \cdot (A-0.4) + W_3\cdot \cos(3\cdot W_1))))). $$
There is no outcome missingness. To mimic a cumulative case-control sampling design, we generate the indicator variable $R$ that satisfies $P(R=1 \mid Y=1) = 1$ and $P(R =1 \mid Y=0) = 0.1$. We choose the threshold $v = median(\text{A})$, which is estimated from a very large sample of $A$. We estimated $\Psi_v^{adj}(P_0)$ using three estimators. The first (Don) is the unadjusted estimator treated in Donovan et al. (2019), the second (Eff) is the fully efficient IPW-srTMLE given in the previous section, the third (IPW) is the inefficient IPW-weighted srTMLE where the inverse of the weights are estimated by empirical means stratified by $Y$. The standard error of each estimator is estimated by computing the estimator on 1000 Monte-Carlo simulations at sample sizes $n = 15000, 25000, 50000$ (which correspond with around 1500, 3300 and 6500 fully-observed observations). The results are plotted below. We see that the fully efficient IPW-srTMLE and the IPW-Donovan et al. estimator have virtually the same standard error. This is expected since the outcome is very rare and thus little gain in efficiency is expected. Since there is no confounding in this simulation, there is no (asymptotic) bias reduction by adjusting for covariates. Interestingly, we see that the inefficient IPW-srTMLE (IPW) has a significantly larger standard error than the Donovan estimator (Don) and fully-efficient IPW-srTMLE (Eff). This shows that adjusting for biomarker missingness using non-targeted IP-weighting can lead to significant losses in statistical efficiency.

\begin{figure}
     \centering
      \includegraphics[width=15cm]{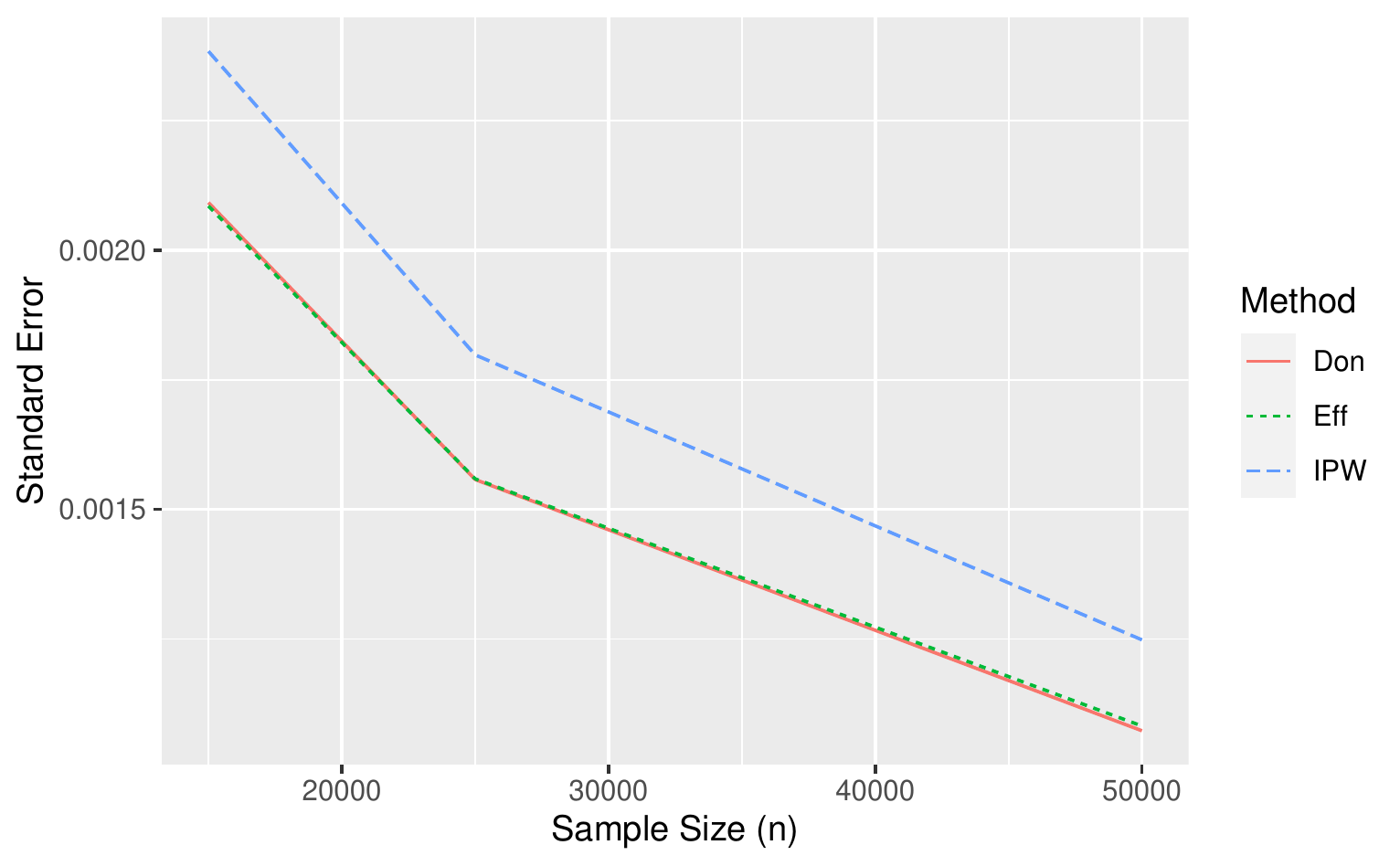}
     \caption{The figure shows a plot of monte-carlo estimated standard errors as a function of sample sizes for the IPW-Donovan estimator (``Don"), fully-efficient IPW-srTMLE (``Eff") and the inefficient IPW-srTMLE (``IPW")}
     \label{fig:my_label}
\end{figure}

\section{Web Appendix F:  Miscellania}
\subsection{Inference for thresholds of zero risk (i.e. absolute protection)}
When a threshold $v$ is such that $P(Y=1|A \geq v) = 0$ then Theorem 1 does not apply. The reason being that the efficient influence function $D_{P,v}$ vanishes:
Recall the efficient influence function is given by $\frac{1(A\geq v)}{P(A\geq v|W)}\left[Y - P(Y=1|A\geq v,W) \right] + P(Y=1|A\geq v, W) - E_W P(Y=1|A\geq v,W)$. If $P(Y=1|A\geq v, W) = 0$ a.e. $W$ then $1(A\geq v)Y = 0$ and $E_W P(Y=1|A\geq v, W) = 0$. But, this implies $\frac{1(A\geq v)}{P(A\geq v|W)}\left[Y - P(Y=1|A\geq v,W) \right] + P(Y=1|A\geq v, W) - E_W P(Y=1|A\geq v,W) = 0$, so that the EIF vanishes.
Because of this, under the assumptions of Theorem 1, the asymptotic component of $\sqrt{n}\left(\Psi_v^{adj}(P_{n,v}^*) - \Psi_v^{adj}(P_0) \right)$ vanishes, and in particular, $\sqrt{n}\left(\Psi_v^{adj}(P_{n,v}^*) - \Psi_v^{adj}(P_0) \right) = o_P(n^{-1/2})$. In fact, for reasonable non-extrapolating estimators, one would actually expect that $\left(\Psi_v^{adj}(P_{n,v}^*) - \Psi_v^{adj}(P_0) \right) $ is identically $0$, since there are no observed events $(Y=1)$ with $A \geq v$. It follows that the TMLE $\Psi_v^{adj}(P_{n,v}^*)$ is faster than $\sqrt{n}$-consistent, however an influence-function-based limit distribution is not available. 

Nonetheless, one can still obtain confidence intervals for the value of the threshold of absolute protection $v$ with $100\%$ finite sample coverage. Suppose that the risk decreases monotonously as the $A$ increases. Let $v_n$ be the maximum immune-response observed for individuals with $Y=1$. By construction, it must be that $P(Y=1|A\geq v_n,W) > 0$ for some $W$ since we observed an individual with $Y=1$ and $A = v_n$. Under the monotonicity constraint, it must be that a threshold of absolute protection (if it exists) is strictly bigger than $v_n$. Thus, $(v_n, \infty)$ is a confidence interval for the threshold of absolute protection that has exact coverage $1$.  

\subsection{Testing the existence of a threshold below a given risk level}
 It is in fact possible to construct a test for the existence of a threshold with risk below a desired cutoff that conservatively controls the type-1 error at level $\alpha$ using our method. To this end, suppose we have simultaneous confidence bands $v \mapsto (l_n(v),u_n(v))$ such that 
$$P \left(\{l_n(v) \leq E_WP(Y=1|A\geq v,W) \leq u_n(v): \forall v\}  \right) = 0.95 + o(1).$$
These simultaneous bands are the boundary of a region that contains the true threshold-response function with probability $0.95 + o(1)$. Thus, any question we answer by assuming that the true threshold-response function falls in this region will be correct at least $95\% + o(1)$ of the time. The question of whether a threshold exists with risk below a given level $\delta$ can be captured by the following null hypothesis:
$$H_0: \forall v,\, E_WP(Y=1|A \geq v,W) > \delta. $$
We propose rejecting the null hypothesis if there exists a threshold $v$ such that the upper simultaneous confidence band satisfies $u_n(v) \leq \delta.$ This in particular implies that every function contained in the confidence region necessary satisfies that its risk is below $\delta$ at the threshold $v$. But, by construction, the true threshold-response function falls in the confidence region with probability $0.95 + o(1)$. Thus, we can conclude that with probability at least $0.95 + o(1)$, the true threshold response function has a risk below $\delta$ at the threshold $v$, and therefore our rejection statistic controls the type-1 error at level $0.05$. A key step of this procedure is that the confidence region is simultaneous, since this allows us to search for a threshold whose risk/upper confidence interval bound is below $\delta$ without losing coverage due to looking at the data. This procedure would not work with pointwise confidence intervals. To achieve better power, one could restrict the region of thresholds $v$ for which to construct the simultaneous confidence bands. For example, one might know apriori that the region of low immune-response will not have risk below $\delta$. In this case, one could choose to only construct the simultaneous confidence bands for thresholds in the upper region of the immune-response. Since the region is smaller, the simultaneous confidence intervals will be less wide. The tradeoff is that your power will be $0$ against alternatives where the threshold of risk below $\delta$ is contained in the lower region, which you no longer look at. 

Alternatively, one could invert the pointwise confidence intervals for the risk of each threshold $v$ into p-values for the point-wise null hypothesis: $H_{0,v}: E_WP(Y=1|A \geq v,W) > \delta$, and then use other multiple testing-based procedures like sequential testing and FDR control to test the collection of null hypothesis': $H_{0,v}:   E_WP(Y=1|A \geq v,W) > \delta$ indexed by thresholds $v$.

\section{Web-Appendix G: Estimation of nuisance parameters and computational considerations}
A key step of the TML procedures for the threshold-response function is the initial estimation of the nuisance parameters $g_{0,v}:= P(A \geq v |W)$,  $Q_{0,v} := E[Y|A\geq v, W]$, and $Q_{0} := E[Y|A = a, W]$. 
$E[Y|A = a, W]$ can be estimated using standard linear or logistic regression methods. For estimating $E[Y|A\geq v, W]$, one procedure is as follows.
\begin{enumerate}
    \item Define weights $w_i = 1(A_i \geq v)$
    \item Perform the weighted regression using $Y_i$ as outcome and $W_i$ as covariates with weights $w_i$
    \item Obtain predicted probabilities
\end{enumerate}
For some regression algorithms, zero weights are not accepted. In this case, one should perform standard unweighted regression using only the observations $O_i$ with $A_i \geq v.$

One way to estimate $P(A \geq v|W)$ directly is as follows.
\begin{enumerate}

    \item Define the pseudo-outcome $\tilde Y_i = 1(A_i \geq v)$
    \item Perform regression using $\tilde Y_i$ as outcome and $W_i$ as covariates
    \item Obtain predicted probabilities
\end{enumerate}

The above estimation methods are reasonable when one is only interested in the risk at a few thresholds. However, recomputing estimates of $P(A\geq v|W)$ and $E[Y|A\geq v,W]$ for a large number of thresholds can be computationally expensive.

Alternatively, one could estimate the full regression function $E[Y|A,W]$ and the conditional density $P(A=da|W):=\frac{d}{da}P(A\leq a|W)$ (or a discrete version $P(v_1 \leq A \leq v_2|W)$ for all consecutive thresholds $v_1, v_2$) instead. Then, one can estimate $Q_{0,v}$ and $g_{0,v}$ through substitution estimators based on
$$Q_{0,v}(W) = \int E_P[Y|A=a,W] \frac{1(a \geq v)}{\int 1(s \geq v)P(A=ds|W)ds}P(A=da|W)da$$
and
$$g_{0,v}(W) = \int 1(a \geq v) P(A=da|W)da.$$
The computation time for estimating $E_P[Y|A,W]$ is not much more than that of $Q_{n,v}$ for a single threshold $v$.  The conditional density $P(A=da|W)$ of $A$ can be estimated using pooled hazard regression methods as described in Diaz, Hejazi (2019)\nocite{D_az_2020}. This is more computationally expensive than estimating $g_{n,v}$ for a single threshold but much cheaper than estimating $g_{n,v}$ for all thresholds.

Another option is to pool the regressions across the thresholds. For example, the following repeated-measures least squares risk function can be minimized to obtain an estimate of $(W,v) \mapsto E_{P_0}[Y|A \geq v,W]$ for the thresholds $\{v_1, v_2, \dots, v_J\}$:
$$R_n(f) = \sum_{j=1}^J\sum_{i=1}^n 1(A_i \geq v_j) \left\{Y_i - f(v_j, W_i) \right\}.$$
By smoothly estimating the dependence on $v$, one extrapolate the estimates to thresholds outside of $\{v_1, v_2, \dots, v_J\}$. For computational efficiency, one can include only a subset of the thresholds of interest in $\{v_1, v_2, \dots, v_J\}$. One issue with this approach is that the estimator might be perform poorly for thresholds $v_j$ that only have a small proportion of the sample above it. We noticed in simulations that this was an issue for poorly calibrated estimators. In practice, it might be helpful to add weights $\frac{1}{P_n(A \geq v_j|W_i)}$ (or $\frac{1}{P_n(A \geq v_j)}$ ) to the above risk function, which will ensure that all thresholds are given equal weight.

\newpage

\section{Web Appendix H:  Code}

\noindent Arguments: 
\begin{itemize}
    \item threshold: a threshold value for A at which to estimate threshold-response
    \item A: a vector of the biomarker observations
    \item Delta: A binary vector for the observed values of the missingness indicator $\Delta$.
    \item Delta: A vector for the observed values of the outcome $Y$.
    \item gv: the evaluation of an estimate of $w \mapsto g_v(w)$ at the observations.
    \item Q: the evaluation of an estimate of $(w,a) \mapsto Q(w,a)$ at the observations.
     \item Qv: the evaluation of an estimate of $(w) \mapsto Q_v(w)$ at the observations.
     \item G: the evaluation of an estimate of $(w,a) \mapsto G(w,a)$ at the observations.
      \item Gv: the evaluation of an estimate of $(w,a) \mapsto G_v(w,a)$ at the observations.
      \item weights: a vector of weights for each observation used for IPW-weighting.
\end{itemize}

\begin{verbatim}

The sequential regression-based TMLE (srTMLE): 

tmle.efficient <- function(threshold, A, Delta, Y,
gv, Q, Qv, G, 
weights = rep(1, length(A)), bound = 0.005) {
  n <- length(A)
  gv <- pmax(gv, bound)
  G <- pmax(G, bound)
  Av <- as.numeric(A >= threshold)
  # Step 1
  eps <- coef(glm.fit(Av, Y, offset = qlogis(Q),
  weights = weights * Delta/(gv*G), 
  family = binomial(), start = 0))
  Q_star <- plogis(qlogis(Q) + eps * Av )
  # Step 2 (sequential regression)
  
  eps <- coef(glm.fit(as.matrix(rep(1,n)), Q_star,
  offset = qlogis(Qv), weights = weights * Av/(gv), 
  family = binomial(), start = 0))
  Qv_star  <- plogis(qlogis(Qv) + eps )
  
  psi <- weighted.mean(Qv_star, weights)
  EIF <- Delta*Av/(gv*G) * (Y - Q_star) + Av/gv * (Q_star - Qv_star) + Qv_star - psi
  EIF <- weights * EIF
  return(list(psi = psi, EIF = EIF))
}


The binary-treatment-based TMLE (binTMLE): 

tmle.inefficient <- function(threshold, A, Delta, Y,
gv, Qv, Gv, 
weights = rep(1, length(A))) {
   n <- length(A)
  Av <- as.numeric(A >= threshold)
  eps <- coef(glm.fit(as.matrix(rep(1,n)), Y,
  offset = qlogis(Qv), weights = weights * Delta*Av/(gv*Gv), 
  family = binomial(), start = 0))
  Qv_star <- plogis(qlogis(Qv) + eps )
  psi <- weighted.mean(Qv_star, weights)
  EIF <- Delta*Av/(gv*Gv) * (Y - Qv_star) + Qv_star - psi
  EIF <- weights * EIF
  return(list(psi = psi, EIF = EIF))
}

\end{verbatim}

\section{Web Appendix I}
The TMLEs presented in this paper can be applied to continuous bounded outcomes as well. A general approach for translating a TMLE for binary outcomes to a TMLE for bounded continuous outcomes is described in Chapter 7 of van der Laan, Rose (2012). We first consider the case when the outcome is contained in the interval $[0,1]$. In this case, there are no changes to the TMLE algorithms, noting that the logistic fluctuation submodels and the risk minimization in the targeting step remain well defined. For general bounded outcomes contained in the interval $[a,b]$, one can transform the outcome by shifting and scaling to a new outcome variable that falls in $[0,1]$. The TMLE for outcomes in $[0,1]$ can then be applied and the resulting estimates and confidence intervals for the estimand associated with the transformed outcome can be transformed back to the original scale.

\end{document}